\documentclass{article}

\usepackage[final]{neurips_2024}

\usepackage[utf8]{inputenc} 
\usepackage[T1]{fontenc}    
\usepackage{hyperref}       
\usepackage{url}            
\usepackage{booktabs}       
\usepackage{amsfonts}       
\usepackage{nicefrac}       
\usepackage{microtype}      
\usepackage{xcolor}         

\usepackage{mathtools}
\usepackage{amssymb, amsmath, amsthm}
\usepackage{bbm}
\hypersetup{colorlinks,citecolor=blue}
\usepackage{cleveref}
\usepackage{bm}
\usepackage{verbatim}
\usepackage{enumitem}
\usepackage{tabularx}
\usepackage{xcolor}
\usepackage{xspace}
\usepackage{todonotes}
\usepackage{amsfonts,mathrsfs,color}
\usepackage{graphicx}
\usepackage{subcaption}
\usepackage{bm}
\usepackage{multirow}
\usepackage{algorithm}
\usepackage{algorithmic}
\newtheorem{definition}{Definition}
\newtheorem{theorem}{Theorem}
\newtheorem*{theorem*}{Theorem}
\newtheorem{lemma}{Lemma}

\newtheorem{corollary}{Corollary}

\newtheorem{assumption}{Assumption}

\newtheorem{claim}{Claim}
\newtheorem{proposition}{Proposition}

\usepackage{pifont}

\DeclareMathOperator*{\argmin}{argmin}
\DeclareMathOperator*{\argmax}{argmax}
\DeclareMathOperator{\poly}{poly}

\DeclareMathOperator{\DualGap}{DualityGap}
\DeclareMathOperator{\KL}{KL}

\newcommand{\notshow}[1]{{}}
\newcommand{\AutoAdjust}[3]{{ \mathchoice{ \left #1 #2  \right #3}{#1 #2 #3}{#1 #2 #3}{#1 #2 #3} }}
\newcommand{\Xcomment}[1]{{}}

\newcommand{\InParentheses}[1]{\AutoAdjust{(}{#1}{)}}
\newcommand{\InBrackets}[1]{\AutoAdjust{[}{#1}{]}}
\newcommand{\InAngles}[1]{\AutoAdjust{\langle}{#1}{\rangle}}
\newcommand{\InNorms}[1]{\AutoAdjust{\|}{#1}{\|}}
\renewcommand{\part}[2]{\frac{\partial #1}{\partial #2}}

\newcommand{\R}{\mathbbm{R}}

\newcommand{\X}{\mathcal{X}}
\newcommand{\Y}{\mathcal{Y}}

\newcommand{\hx}{\widehat{x}}
\newcommand{\hy}{\widehat{y}}
\newcommand{\hz}{\widehat{z}}

\newcommand{\tx}{\tilde{x}}
\newcommand{\ty}{\tilde{y}}

\newcommand{\supp}{\mathrm{supp}}

\newcommand{\defeq}{\coloneqq}
\newcommand{\rmp}{RM$^+$}

\def\+#1{\mathcal{#1}}
\def\-#1{\mathbb{#1}}

\usetikzlibrary{patterns,shapes.geometric,patterns.meta,decorations.pathmorphing}

\title{Fast Last-Iterate Convergence of Learning in Games Requires Forgetful Algorithms}

\author{%
  Yang Cai\\
  Yale\\
  \texttt{yang.cai@yale.edu} \\
  \And
  Gabriele Farina\\
  MIT\\
  \texttt{gfarina@mit.edu} \\
  \And
  Julien Grand-Clément\\
  HEC Paris\\
  \texttt{grand-clement@hec.fr} \\
  \And
  Christian Kroer\\
  Columbia\\
  \texttt{ck2945@columbia.edu}\\
  \And
  Chung-Wei Lee\\
  USC\\
  \texttt{leechung@usc.edu}\\
  \And
  Haipeng Luo\\
  USC\\
  \texttt{haipengl@usc.edu}\\
  \And
  Weiqiang Zheng\\
  Yale\\
  \texttt{weiqiang.zheng@yale.edu}
}

\begin{document}

\maketitle

\begin{abstract}
    Self-play via online learning is one of the premier ways to solve large-scale two-player zero-sum games, both in theory and practice. Particularly popular algorithms include optimistic multiplicative weights update (OMWU) and optimistic gradient-descent-ascent (OGDA). While both algorithms enjoy $O(1/T)$ ergodic convergence to Nash equilibrium in two-player zero-sum games, OMWU offers several advantages including logarithmic dependence on the size of the payoff matrix and $\Tilde{O}(1/T)$ convergence to coarse correlated equilibria even in general-sum games.
    However, in terms of last-iterate convergence in two-player zero-sum games, an increasingly popular topic in this area, OGDA guarantees that the duality gap shrinks at a rate of $(1/\sqrt{T})$, while the best existing last-iterate convergence for OMWU depends on some game-dependent constant that could be arbitrarily large.
    This begs the question: is this potentially slow last-iterate convergence an inherent disadvantage of OMWU, or is the current analysis too loose?
    Somewhat surprisingly, we show that the former is true.
    More generally, we prove that a broad class of algorithms that do not forget the past quickly all suffer the same issue: for any arbitrarily small $\delta>0$, there exists a $2\times 2$ matrix game such that the algorithm admits a constant duality gap even after $1/\delta$ rounds.
    This class of algorithms includes OMWU and other standard optimistic follow-the-regularized-leader algorithms.
\end{abstract}


\section{Introduction}
Self-play via online learning is one of the premier ways to solve large-scale two-player zero-sum games. Major examples include super-human AIs for Go, Poker~\citep{brown2018superhuman}, and human-level AI for Stratego~\citep{perolat2022mastering} and alignment of large language models~\citep{munos2023nash}.
In particular, Optimistic Multiplicative Weights Update (OMWU) and Optimistic Gradient Descent-Ascent (OGDA) are two of the most well-known online learning algorithms.
When applied to learning a two-player zero-sum game via self-play for $T$ rounds, the {\em average} iterates of both algorithms are known to be an $O(1/T)$-approximate Nash equilibrium~\citep{rakhlin2013optimization, syrgkanis2015fast}, while other algorithms, such as vanilla Multiplicative Weights Update (MWU) and vanilla Gradient Descent-Ascent (GDA), have a slower ergodic convergence rate of $O(1/\sqrt{T})$.

For multiple practical reasons, there is growing interest in studying the {\em last-iterate} convergence of these learning dynamics~\citep{daskalakis2019last, golowich2020last, wei2021linear,lee2021last}.
In this regard, existing results seemingly exhibit a gap between OGDA and OMWU --- the duality gap of the last iterate of OGDA is known to decrease at a rate of $O(1/\sqrt{T})$~\citep{cai2022finite, gorbunov2022last}, with no dependence on constants beyond the dimension and the smoothness of the players' utility functions of the game.\footnote{In finite two-player zero-sum games, the dependence is polynomial in the number of actions and the largest absolute value in the payoff matrix.} In contrast, the existing convergence rate for OMWU depends on some game-dependent constant that could be arbitrarily large,  even after fixing the dimension and the smoothness constant of the game~\citep{wei2021linear}.\footnote{We note that there are also linear-rate last-iterate results for OGDA when we allow dependence on such constants; see~\citep{wei2021linear}.}
Given the fundamental role of OMWU in online learning and its other advantages over OGDA (such as its logarithmic dependence on the number of actions), it is natural to ask the following question:
\begin{equation}
    \textit{Is the potentially slow last-iterate convergence an inherent disadvantage of OMWU?} \tag{*}
\end{equation}

\paragraph{Main Results.}
In this work, we show that the answer to this question is yes,
contrary to a common belief that better analysis and better last-iterate convergence results similar to those of OGDA are possible for OMWU. More specifically, we show the following.

\begin{theorem*}[Informal]
    For OMWU with constant step size, there is no function $f$ such that the corresponding learning dynamics $\{(x^t,y^t)\}_{t\geq 1}$ in two-player zero-sum games $[0,1]^{d_1 \times d_2}$ has a last-iterate convergence rate of $f(d_1,d_2, T)$.\footnote{Under the same condition, OGDA has a last-iterate convergence rate of $\frac{\poly(d_1,d_2)}{\sqrt{T}}$.} More specifically, no function $f$ can satisfy
    \begin{itemize}
        \item[1.] $\DualGap(x^T, y^T) \le f(d_1,d_2, T)$ for all matrices $[0,1]^{d_1 \times d_2}$ and $T \ge 1$.
        \item[2.] $\lim_{T\rightarrow \infty}f(d_1,d_2, T) \rightarrow 0$.
    \end{itemize}
\end{theorem*}

Our findings show that, despite the significantly superior \emph{regret} properties of OMWU compared to OGDA, its \emph{last-iterate convergence} properties are remarkably worse. In turn, this counters the viewpoint that ``Follow-the-Regularized-Leader (FTRL) is better than Online Mirror Descent (OMD)''~\citep{erven2021why}:
crucially, while OMWU is an instance of (optimistic) FTRL, OGDA is an instance of optimistic OMD that cannot be expressed in the FTRL formalism.

We further show that similar negative results extend to several other standard online learning algorithms, including a close variant of OGDA.
More concretely, our main results are as follows.

\begin{figure}[t]
    \includegraphics[scale=.72]{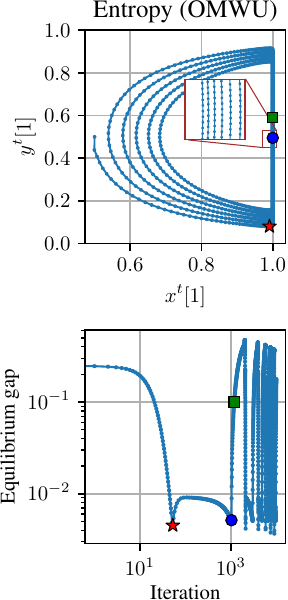}\hfill%
    \includegraphics[scale=.72]{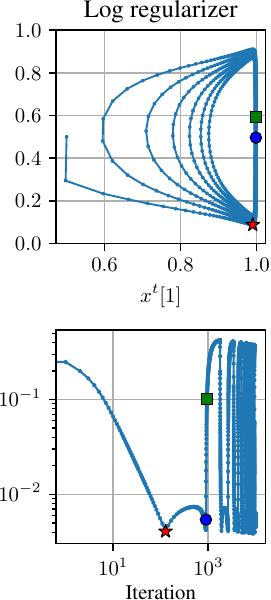}\hfill%
    \includegraphics[scale=.72]{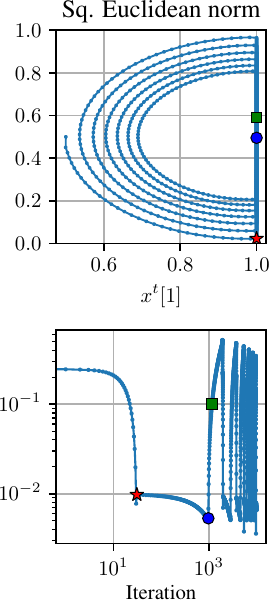}\hfill%
    \includegraphics[scale=.72]{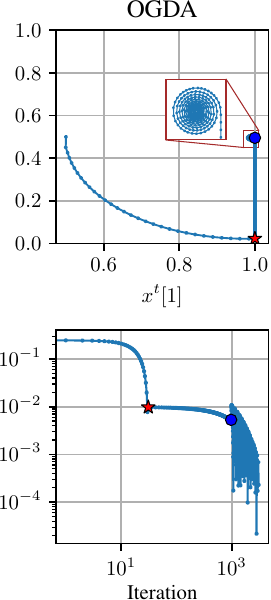}%
    \caption{Comparison of the dynamics produced by three variants of OFTRL with different regularizers (negative entropy, logarithmic regularizer, and squared Euclidean norm) and OGDA in the same game $A_\delta$ defined in \eqref{eq:A delta} for $\delta := 10^{-2}$. The bottom row shows the duality gap achieved by the last iterates.
    The OFTRL variants exhibit poor performance due to their lack of \emph{forgetfulness}, while OGDA converges quickly to the Nash equilibrium.
    Since the regularizers in the first two plots are Legendre, the dynamics are equivalent to the ones produced by optimistic OMD with the respective Bregman divergences.
    In the plot for OMWU we observe that $x^t[1]$ can get extremely close to the boundary (\emph{e.g.,} in the range $1 - e^{-50} < x^t[1] < 1$).
    To correctly simulate the dynamics, we used
    1000 digits of precision. The red star, blue dot, and green square illustrate the key times $T_1$, $T_2$, $T_3$ defined in our analysis in Section~\ref{sec:analysis}.
    }
    \label{fig:intro plots}
\end{figure}

\begin{figure}[htp]
    \includegraphics[width=\linewidth]{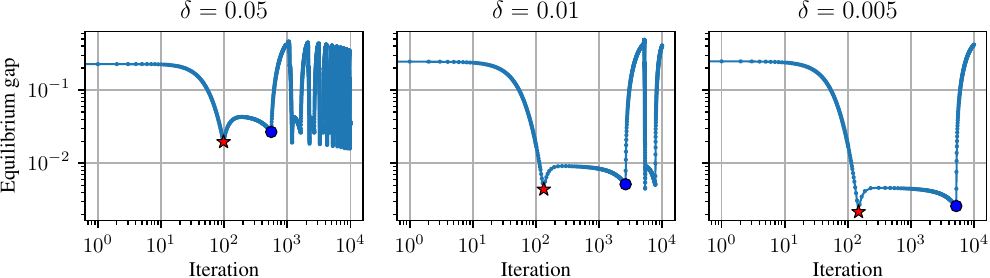}
    \caption{Performance of OMWU on the game $A_\delta$ defined in \cref{eq:A delta} for three choices of $\delta$. In all plots, the learning rate was set to $\eta = 0.1$. As predicted by our analysis, the length of the ``flat region'' between iteration $T_1$ (red star) and $T_2$ (blue dot) scales inversely proportionally with $\delta$.}
    \label{fig:delta effect}
\end{figure}

\begin{itemize}[leftmargin=*]
    \item We identify a broad family of Optimistic FTRL (OFTRL) algorithms that do not forget about the past quickly. We prove that, for any sufficiently small $\delta>0$, there exists a $2\times 2$ two-player zero-sum game such that, even after $1/\delta$ iterations, the duality gap of the iterate output by these algorithms is still a constant (\Cref{theorem: main}).
          This excludes the possibility of showing a game-independent last-iterate convergence rate similar to that of OGDA.

    \item We prove that many standard online learning algorithms, such as OFTRL with the entropy regularizer (equivalently, OMWU), the Tsallis entropy family of regularizers,
          the log regularizer, and the squared Euclidean norm regularizer, all fall into this family of non-forgetful algorithms and thus all suffer from the same slow convergence.
          Also note that Optimistic OMD (OOMD), another well-known family of algorithms, is equivalent to OFTRL when given a Legendre regularizer. Therefore, OOMD with the entropy, Tsallis entropy, and log regularizer also suffer the same issue.\footnote{We focus on optimistic variants of these algorithms since it is well-known that their vanilla version does not converge in the last iterate at all, see e.g.~\citep{mertikopoulos2018cycles, daskalakis2018limit,bailey2018multiplicative,cheung2019vortices}.}

    \item Finally, we also generalize our negative results from $2\times 2$ games to $2n\times 2n$ games for any positive integer $n$, strengthening our message that forgetfulness is generally needed in order to achieve fast last-iterate convergence.
\end{itemize}

\paragraph{Main Ideas.} Intuitively, we trace the poor last-iterate convergence properties of OFTRL to its \emph{lack of forgetfulness}.
The high-level idea of our hard $2\times 2$ game instance, parametrized by $\delta>0$, is as follows. First, it has a unique Nash equilibrium at which one player is $O(\delta)$ close to the boundary of the simplex. We refer to the first row of plots in \Cref{fig:intro plots}, where the equilibrium is noted by a blue dot (note that we can plot only $x[1],y[1]$ for each player, since $x[2]=1-x[1]$ and $y[2]=1-y[2]$). As can be seen, the iterates of OGDA and all three OFTRL variants initially have a two-phase structure. In the first phase, they converge to the lower-right area denoted by a red star in \Cref{fig:intro plots}. Then, from there all algorithms start moving towards the equilibrium. In particular, $y[1]$ increases. However, once they enter the vicinity of the equilibrium, the behavior depends on the algorithms. For OGDA, the dynamics start spiraling closer and closer to the equilibrium. On the other hand, for the OFTRL algorithm, the $x$ player has built up a lot of ``memory" of $x[1]$ being better than $x[2]$, and for this reason, $x[1]$ will stay very close to $1$ for a long time. During the time when $x[1]$ is close to $1$, $y[1]$ keeps increasing since the $y$ player receives gradients that indicates $y[1]$ is better than $y[2]$. As a result, the dynamics cannot ``stop'' near the equilibrium but start to move away from the equilibrium. The dynamics reach a point (denoted by a green square) whose duality gap is a constant and enter a new cycle where they move out towards the starting point of the learning process. This cycle repeats in smaller and smaller semi-ellipses that slowly converge to equilibrium. Note that the semi-ellipses correspond to the seesaw pattern in the equilibrium gap (second row of plots). OFTRL overshoots the equilibrium as it has built up a lot of "memory" of $x[1]$ being better than $x[2]$ along the phase from the red star to the blue circle, and it requires many iterations to "forget" this fact. We show that as we make $\delta$, the parameter defining the nearness to the boundary, smaller and smaller, it takes longer and longer for these semi-ellipses to get close to the equilibrium along the entire path, as illustrated in~\Cref{fig:delta effect}.



Our results are related to numerical observations made in the literature on solving large-scale extensive-form games. There, algorithms based on the regret-matching$^+$ (\rmp) algorithm~\citep{tammelin2015solving}, combined with the counterfactual regret minimization~\citep{zinkevich2007regret}, perform by far the best in practice. In contrast, the classical regret matching algorithm~\citep{hart2000simple} performs much worse, in spite of similar regret guarantees. It was later discovered that \rmp{} corresponds to OGD, while RM corresponds to FTRL~\citep{farina2021faster,flaspohler2021online}. It was hypothesized that RM builds up too much negative regret at times, and thus is slow to adapt to changes in the learning dynamics related to the strategy of the other player.
These numerical results and the hypothesis are consistent with our theoretical findings: FTRL (and thus RM) is not able to ``forget,'' whereas OGD and OGDA can forget and thereby quickly adapt to changes in which actions should be played.



\subsection{Related Work}

The literature on last-iterate convergence of online learning methods in games is vast. In this section, we will cover key contributions focusing on the case of interest for this paper: discrete-time dynamics for two-player zero-sum normal-form games.

\emph{Convergence of OGDA.}
Average-iterate convergence of OGDA has been studied for minimax optimization problems in both the unconstrained~\citep{mokhtari2020convergence} and constrained settings~\citep{hsieh2019convergence}. Last-iterate convergence of OGDA in \emph{unconstrained} saddle-point problems has been shown in \citep{daskalakis2018training, golowich2020tight}.
In the (constrained) game setting,
\citet{wei2021linear, anagnostides2022last} showed \emph{best}-iterate convergence to the set of Nash equilibria in any two-player zero-sum game with payoff matrix $A$ at a rate of $O(\textrm{poly}(d_1, d_2, \max_{i,j} |A_{i,j}|)/\sqrt T)$
using constant learning rate, where $d_1$ and $d_2$ are the number of actions of the players. A stronger result was shown by \citet{cai2022finite}, who showed
that the same rate applies to the \emph{last} iterate.

\emph{Convergence of OMWU.} Optimistic multiplicative weights update (also known as optimistic hedge) is often regarded as the premier algorithm for learning in games. Unlike OGDA, it guarantees sublinear regret with a \emph{logarithmic} dependence on the number of actions, and it is known to guarantee only polylogarithmic regret per player when used in self play even for general-sum games~\citep{daskalakis2021near-optimal}. It can be applied with similar strong properties beyond normal-form games in several important combinatorial settings~\citep{takimoto2003path,koolen2010hedging,farina22:kernelized}. The work by \citet{daskalakis2019last} established \emph{asymptotic} last-iterate convergence for OMWU in games using a small learning rate under the assumption of a unique Nash equilibrium. Similar asymptotic results without the unique equilibrium assumption were also given by \citet{mertikopoulos2019optimistic,hsieh2021adaptive}. \citet{wei2021linear} were the first to provide \emph{nonasymptotic} learning rates for OMWU. Specifically, they showed a linear rate of convergence in games with a unique equilibrium, albeit with a dependence on a condition number-like quantity that could be arbitrarily large given fixed $d_1$, $d_2$, and $\max_{i,j}|A_{i,j}|$.
This result was later extended by \citet{lee2021last} to extensive-form games. Unlike OGDA, no last-iterate convergence result for OMWU with a polynomial dependence on only the natural parameters of the game (\textit{i.e.}, $d_1$, $d_2$, and $\max_{i,j}|A_{i,j}|$) is known. As we show in this paper, perhaps surprisingly, this is no coincidence: in general, OMWU does not exhibit a last-iterate convergence rate that solely depends on these parameters, whether polynomial or not.

\emph{FTRL vs. OMD.}
While the last-iterate convergence of instantiations of Optimistic Online Mirror Descent has been observed before, the properties of Follow-the-Regularized-Leader dynamics remain mostly elusive. The present paper partly explains this vacuum: all standard instantiations of optimistic FTRL \emph{cannot hope} to converge in iterates with only a polynomial dependence on the natural parameters of the game, unlike optimistic OMD.
Complications in obtaining last-iterate convergence results for continuous-time FTRL instantiations were already reported by \citet{vlatakis2020no}, who showed the necessity of \emph{strict} Nash equilibria. 

\emph{Exploiting a no-regret learner.}
The forgetfulness property that we identify is closely related to the concept of \emph{mean-based} learning algorithms from \citet{braverman2018selling}. Intuitively, mean-based algorithms are ones such that if the mean reward for action $a$ is significantly greater than the mean reward for action $b$, then the algorithm selects $b$ with negligible probability.
They show that MWU is mean-based, along with Follow-the-Perturbed-Leader and the Exp3 bandit algorithm.
\citet{braverman2018selling} shows that "mean-based" algorithms are exploitable when learning to bid in first-price auctions, whereas \citet{kumar2024strategically} shows that OGD does not suffer from this exploitability issue.


\section{Preliminaries and Problem Setup}
\label{sec: pre}
We consider the standard setting of no-regret learning in a zero-sum game $A \in [0,1]^{d_1 \times d_2}$. In each iteration $t \ge 1$, the $x$-player chooses $x^t \in \+X := \Delta^{d_1}$ while the $y$-player chooses $y^t \in \+Y:= \Delta^{d_2}$.
Then the $x$-player receives loss vector $\ell^t_x = Ay^t$ while the $y$-player receives loss vector $\ell^t_y = -A^\top x^t$.
The goal is to find or approximate a \emph{Nash equilibrium} $(x^*,y^*)$ to the game such that $x^*\in\argmin_{x\in\+X}\max_{y\in\+Y}x^\top Ay$ and $y^*\in\argmax_{y\in\+Y}\min_{x\in\+X}x^\top Ay$.
The approximation error of a strategy pair $(x, y)$ is measured by its duality gap, defined as $\DualGap(x,y) = \max_{y'\in \+Y} x^\top Ay' - \min_{x' \in \+X} {x'}^\top A y$, which is always non-negative.

Popular no-regret algorithms for solving the game include the Optimistic Follow-the-Regularized-Leader (OFTRL) algorithm and the Optimistic Online Mirror Descent (OOMD) algorithm, both defined in terms of a certain regularizer $R: \Delta^d \rightarrow \-R$ (for some general dimension $d$). The corresponding Bregman divergence of $R$ is $D_R(x,x') = R(x) - R(x') - \InAngles{\nabla R(x'), x- x'}$, and the regularizer is $1$-strongly convex if $D_R(x,x') \ge \frac{1}{2}\InNorms{x -x'}_2^2$ for all $x, x' \in \Delta^d$.

\paragraph{Optimistic Online Mirror Descent (OOMD)}
Starting from an initial point $ (x^1, y^1) =  (\hx^1, \hy^1)$, the OOMD algorithm with regularizer $R$ and steps size $\eta > 0$ updates in each iteration $t \ge 2$,
\begin{equation}
    \label{OOMD}
    \tag{OOMD}
    \begin{aligned}
        \hx^{t} & = \argmin_{x \in \X} \{\eta \InAngles{x, \ell^{t-1}_x} + D_R(x, \hx^{t-1})\}, \quad  x^{t} = \argmin_{x \in \X} \{\eta \InAngles{x, \ell^{t-1}_x} + D_R(x, \hx^{t})\},  \\
        \hy^{t} & = \argmin_{y \in \Y} \{\eta \InAngles{y, \ell^{t-1}_y} + D_R(y, \hy^{t-1})\}, \quad  y^{t} = \argmin_{y \in \Y} \{\eta \InAngles{y, \ell^{t-1}_y} + D_R(y, \hy^{t})\}.
    \end{aligned}
\end{equation}
In particular, we call OOMD with a squared Euclidean norm regularizer, that is, $R(x)=\frac{1}{2}\sum_{i=1}^dx[i]^2$ \emph{optimistic gradient-descent-ascent} (OGDA).
When $R$ is the negative entropy, that is, $R(x)=\sum_{i=1}^{d}x[i]\log x[i]$, we call the resulting OOMD algorithm \emph{optimistic multiplicative weights update} (OMWU).
OGDA and OMWU have been extensively studied in the literature regarding their last-iterate convergence properties in zero-sum games.
Specifically, both OMWU and OGDA guarantee that $(x^t,y^t)$ approaches to a Nash equilibrium as $t\to \infty.$

\paragraph{Optimistic Follow-the-Regularized-Leader (OFTRL)}
Define the cumulative loss vectors $L^t_x := \sum_{k=1}^t \ell^k_x$ and $L^t_y := \sum_{k=1}^t \ell^k_y$. 
The update rule of OFTRL with regularizer $R$ is for each $t \ge 1$,
\begin{equation}
    \label{OFTRL}
    \tag{OFTRL}
    \begin{aligned}
        x^t & = \argmin_{x \in \X} \left\{  \InAngles{x, L^{t-1}_x + \ell^{t-1}_x } + \frac{1}{\eta} R(x)\right\}, \\
        y^t & = \argmin_{y \in \Y} \left\{  \InAngles{y, L^{t-1}_y + \ell^{t-1}_y } + \frac{1}{\eta} R(y)\right\}.
    \end{aligned}
\end{equation}
Throughout the paper, we consider the following regularizers:
\begin{itemize}[leftmargin=*]
    \item Negative entropy ($R(x)=\sum_{i=1}^{d}x[i]\log x[i]$): the resulting OFTRL algorithm coincides with OMWU defined by the OOMD framework previously.
    \item Squared Euclidean norm ($R(x)=\frac{1}{2}\sum_{i=1}^dx[i]^2$): note that the resulting algorithm is different from OGDA since the squared Euclidean norm is not a Legendre regularizer. As we will show, the two algorithms behave very differently in terms of last-iterate convergence.
    \item Log barrier ($R(x)=\sum_{i=1}^d-\log (x[i])$): we also call it the log regularizer.
    \item Negative Tsallis entropy regularizers ($R(x) = \frac{1 - \sum_{i=1}^d (x[i])^\beta}{1-\beta}$ parameterized by $\beta \in (0,1)$).
\end{itemize}


\paragraph{The 2-dimension case}
We denote $x \in \R^2$ as $x = [x[1], x[2]]^\top$. For $d_1=2$, finding $x^t$ of OFTRL reduces to the following 1-dimensional optimization problem:
\begin{equation*}
    x^t[1] = \argmin_{x \in [0,1]} \left\{ x \cdot (L^{t-1}_x[1] + \ell^{t-1}_x[1] - L^{t-1}_x[2] - \ell^{t-1}_x[2]) + \frac{1}{\eta} R(x)  \right\}, \quad x^t[2] = 1 - x^t[1],
\end{equation*}
where we slightly abuse the notation and denote $R(x) = R([x, 1-x])$ for $x \in [0,1]$. We introduce two notations (the case for the $y$-player is similar):
let $e^t_x = \ell^t_x[1] - \ell^t_x[2]$ be the difference between the losses of the two actions, and $E^t_x = \sum_{k=1}^t e^k_x$ be the cumulative difference between the losses of the two actions.
For OFTRL, it is clear that the update of $x^t$ only depends on the differences $E^{t-1}_x, e^{t-1}_x$, the step size $\eta$, and the regularizer $R$. For this reason, we define $F_{\eta, R} : \-R \rightarrow [0,1]$ as follows:
\begin{align}
    \label{function:F}
    F_{\eta, R}(e) := \argmin_{x \in [0,1]}\left\{ x \cdot e + \frac{1}{\eta} R(x) \right\}.
\end{align}
We assume the function $F_{\eta, R}$ is well-defined, \textit{i.e.}, the above optimization problem admits a unique solution in $[0,1]$. This is a condition easily satisfied, for example, when the regularizer $R$ is strongly convex. Then the OFTRL algorithm can be written as
\begin{align*}
    x^t[1] = F_{\eta, R}(E^{t-1}_x + e^{t-1}_x), \quad x^t[2] = 1 - x^t[1].
\end{align*}
The following lemma shows that the function $F_{\eta, R}$ is non-increasing (we defer missing proofs in the section to \Cref{app:pre}).
\begin{lemma}[Monotonicity of $F_{\eta, R}$]
    \label{lemma: F monotone}
    The function $F_{\eta, R}(\cdot): \-R \rightarrow [0,1]$ defined in \eqref{function:F} is non-increasing.
\end{lemma}

We present some blanket assumptions on the regularizer, which are satisfied by all the regularizers introduced before.
\begin{assumption}
    \label{assumption:standard}
    We assume that the regularizer $R$ satisfies the following properties: the function $F_{\eta, R}: \-R \rightarrow [0,1]$ defined in \eqref{function:F} is, 
    \begin{itemize}
        \item[1.] \textbf{Unbiased:} $F_{\eta, R}(0) = \frac{1}{2}$.
        \item[2.] \textbf{Rational:} $\lim_{E \rightarrow -\infty}F_{\eta, R}(E) = 1$ and $\lim_{E \rightarrow +\infty}F_{\eta, R}(E) = 0$.
        \item[3.] \textbf{Lipschitz continuous:} There exists $L \ge 0$ such that $F_{1, R}$ is $L$-Lipschitz.

    \end{itemize}
\end{assumption}
Item 1 in \Cref{assumption:standard} shows that the initial strategy is the uniform distribution over the two actions, which is standard in practice. The rational assumption (item 2 in \Cref{assumption:standard}) is natural since otherwise, the algorithm could not even converge to a pure Nash equilibrium. The Lipschitzness (item 3 in \Cref{assumption:standard}) is implied when the regularizer is strongly convex over $[0,1]^2$ (see \Cref{lemma:Lipschitz}),
and it further implies Lipschitzness of $F_{\eta, R}$ for any $\eta$ as shown in the following proposition.
\begin{proposition}
    \label{prop: lip}
    The function $F_{\eta, R}$ satisfies $F_{\eta, R}\InParentheses{E / \eta} = F_{1, R}(E)$. If $F_{1, R}$ is $L$-Lipschitz, then $F_{\eta, R}$ is $\eta L$-Lipschitz for  any $\eta > 0$.
\end{proposition}

\section{Slow Convergence of OFTRL: A Hard Game Instance}\label{sec:analysis}

We give negative results on the last-iterate convergence properties of OFTRL by studying its behavior on a surprisingly simple $2 \times 2$ two-player zero-sum games. The game's loss matrix $A_{\delta}$ is parameterized by $\delta\in(0,\frac{1}{2})$ and is defined as follows:
\begin{equation}\label{eq:A delta}
    A_\delta \defeq \begin{bmatrix}
        \frac{1}{2} + \delta & \frac{1}{2} \\[2mm]
        0                    & 1
    \end{bmatrix}.
\end{equation}

\subsection{Basic Properties}
We summarize some useful properties of $A_\delta$ in the following proposition (missing proofs of this section can be found in \Cref{app:analysis}).
\begin{proposition}\label{proposition:properties of A}
    The matrix game $A_\delta$ satisfies:
    \begin{itemize}
        \item[1.] $A_\delta$ has a unique Nash equilibrium $x^* = [  \frac{1}{1+\delta}, \frac{\delta}{1 + \delta}]$ and $y^* = [ \frac{1}{2(1+\delta)}, \frac{1 + 2\delta}{2(1 + \delta)}]$.
        \item[2.] For a strategy pair $(x^t, y^t)$, the loss vectors (\textit{i.e.}, gradients) for the two players are respectively:
            \begin{align}
                \ell^t_x = A_\delta y^t =
                \begin{bmatrix}
                    \frac{1}{2} + \delta y^t[1] \\
                    1 - y^t[1]
                \end{bmatrix}
                \quad
                \ell^t_y = -A_\delta^\top x^t =
                -\begin{bmatrix}
                     (\frac{1}{2} + \delta)x^t[1] \\
                     1 - \frac{1}{2}x^t[1]
                 \end{bmatrix}.
            \end{align}
            Moreover,
            \begin{align*}
                 & e^t_x = \ell^t_x[1] -\ell^t_x[2] = -\frac{1}{2} + (1+\delta)y^t[1] \in [-\frac{1}{2}, \frac{1}{2} + \delta] \\
                 & e^t_y = \ell^t_y[1] -\ell^t_y[2] =  1 -  (1+\delta) x^t[1] \in [-\delta, 1].
            \end{align*}
    \end{itemize}
\end{proposition}
In particular, we notice that $e^t_y \ge -\delta$. It implies that if the cumulative differences between the losses of the two actions $E^t_y$ is large, then it takes $\Omega(\frac{1}{\delta})$ iterations to make $E^t_y$ small (close to $0$). This has important implications for non-forgetful algorithms like OFTRL that look at the whole history of losses. Since OFTRL chooses the strategy $y^t$ based on $E^t_y$, it could be trapped in a bad action for a long time even if the current gradients suggest that the other action is better. This is the key observation for our main negative results on the slow last-iterate convergence rates of OFTRL.

The following lemma shows that in a particular region of $(x, y)$, the duality gap is a constant. 
\begin{lemma}\label{lem:duality gap is large}
    Let $\delta, \epsilon \in (0, \frac{1}{2})$. For any $x,y \in \Delta^2$ such that $x[1] \geq \frac{1}{1+\delta}$ and $y[1] \geq \frac{1}{2}+\epsilon$, the duality gap of $(x, y)$ for game $A_\delta$ (defined in \eqref{eq:A delta}) satisfies
    $\DualGap(x,y) \geq \epsilon.$
\end{lemma}

\subsection{Slow Last-Iterate Convergence}
We further require the following assumption on the regularizer $R$ (and thus the function $F_{1,R}$).
\begin{assumption}\label{assumption:main}
    Let $L$ be the Lipschitness constant of $F_{1,R}$ in \Cref{assumption:standard}. Denote constant $c_1 = \frac{1}{2} - F_{1, R}(\frac{1}{20L})$. There exist universal constants $\delta', c_2 > 0$ and $c_3 \in (0,\frac{1}{2}]$ such that for any $ 0 < \delta \le \delta'$,
    \begin{itemize}
        \item[1.] 
            For any $E$ that satisfies $F_{1, R}(E) \ge \frac{1}{1+\delta}$, we have $F_{1, R}(-\frac{c_1^2}{30 L\delta} + E) \ge \frac{1 + c_3}{1 + c_3 + \delta}$
        \item[2.] 
             For any $E$ that satisfies $F_{1, R}(E) \ge \frac{1}{2(1+\delta)}$, we have $F_{1, R}(-\frac{c_3 c_1^2}{120L} + \frac{\delta}{4L} + E) \ge \frac{1}{2} + c_2$.
    \end{itemize}
\end{assumption}
Although \Cref{assumption:main} is technical, the idea is simple. Item 1 in \Cref{assumption:main} states that if a loss difference $E < 0$ already makes $F_{1,R}(E) \ge \frac{1}{1+\delta}$, then the loss difference $E' = E - \Omega(\frac{1}{\delta})$ is able to make $F_{1,R}(E')$ greater than $F_{1,R}(E)$ by a margin of $\Omega(\delta)$. Item 2 in \Cref{assumption:main} states that if a loss difference $E$ already makes $F_{1,R}(E) \ge \frac{1}{2(1+\delta)} \approx \frac{1}{2}$, then the loss difference $E' = E - \Omega(1)$ is able to make $F_{1,R}(E')$ greater than $\frac{1}{2}$ by a constant margin $c_2$. In \Cref{app:verification}, we verify that \Cref{assumption:main} holds for the negative entropy, squared Euclidean norm, the log barrier, and the negative Tsallis entropy regularizers.

Now we present the main result of the section showing that even after $\Omega(1/\delta)$ iterations, the duality gap of the iterate output by OFTRL is still a constant.
\begin{theorem}
    \label{theorem: main}
    Assume the regularizer $R$ satisfies \Cref{assumption:standard} and \Cref{assumption:main}. For any $\delta \in (0,\hat{\delta})$, where $\hat{\delta}$ is a constant depending only on the constants $c_1$ and $\delta'$ defined in \Cref{assumption:main},
    the OFTRL dynamics on $A_\delta$ (defined in \eqref{eq:A delta}) with any step size $\eta \le \frac{1}{4L}$ satisfies the following: there exists an iteration $t \ge \frac{c_1}{3\eta L \delta}$ with a duality gap of at least $c_2$,  a strictly positive constant defined in \Cref{assumption:main}.
\end{theorem}

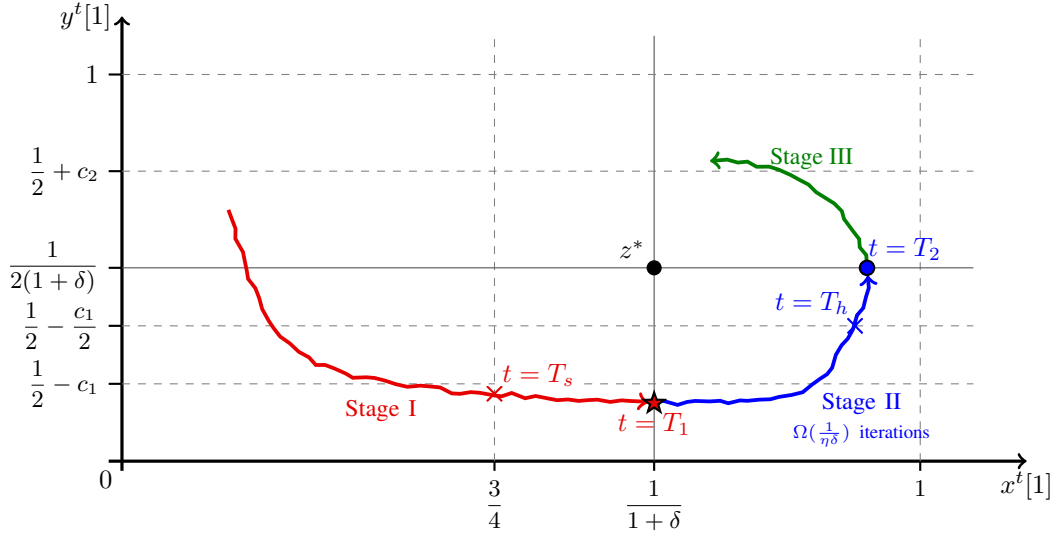
\begin{figure}[t]
    \centering
    \begin{tikzpicture}[x=3.52cm,y=2.56cm,decoration={random steps,segment length=4pt,amplitude=1pt}, scale=1]
        \tikzset{cross/.style={path picture={
                            \draw[thick]
                            (path picture bounding box.south east) -- (path picture bounding box.north west) (path picture bounding box.south west) -- (path picture bounding box.north east);
                        }}}
        \draw[very thick,->] (0,-.05) -- (0,2.3) node[left] {$y^t[1]$};
        \draw[very thick,->] (-.05,0) -- (3.4,0) node[below] {$x^t[1]$};
        \path (0,0) node[below left] {$0$};
        \draw[thick] (2,0) -- +(0,-.05) node[below] {\small$\displaystyle\frac{1}{1+\delta}$};
        \draw[thick] (3,0) -- +(0,-.05) node[below] {\small$1$};
        \draw[thick] (1.4,0) -- +(0,-.05) node[below] {\small$\displaystyle\frac 3 4$};
        \draw[thick] (0,1) -- +(-.05,0) node[left] {\small$\displaystyle\frac{1}{2(1+\delta)}$};
        \draw[thick] (0,.7) -- +(-.05,0) node[left] {\small$\displaystyle\frac{1}{2}-\frac{c_1}{2}$};
        \draw[thick] (0,.4) -- +(-.05,0) node[left] {\small$\displaystyle\frac{1}{2}-c_1$};
        \draw[thick] (0,1.5) -- +(-.05,0) node[left] {\small$\displaystyle\frac{1}{2}+c_2$};
        \draw[thick] (0,2) -- +(-.05,0) node[left] {$\displaystyle 1$};
        \draw[thin,dashed,gray] (3,0) -- +(0,2.2);
        \draw[black!70] (2,0) -- +(0,2.2);
        \draw[thin,dashed,gray] (1.4,0) -- +(0,2.2);
        \draw[black!70] (0,1) -- +(3.2,0);
        \draw[thin,dashed,gray] (0,2) -- +(3.2,0);
        \draw[thin,dashed,gray] (0,.7) -- +(3.2,0);
        \draw[thin,dashed,gray] (0,.4) -- +(3.2,0);
        \draw[thin,dashed,gray] (0,1.5) -- +(3.2,0);
        \fill[black] (2,1) circle(.1cm) node[above left] {$z^*$};

        \draw[decorate,line width=.5mm,red!90!black,->,shorten >=.5mm] (.4,1.3) .. controls (.6,.4) .. node[pos=.65,below=1mm,align=center,text width=2.5cm]{\small Stage I} (1.99,.3);
        \draw[decorate,line width=.5mm,blue,->,shorten >=.5mm] (2,.3) .. controls (2.5, .3) .. node[pos=0.9,xshift=6mm,below,text width=2.5cm,align=center]{\small Stage II\\\scriptsize $\Omega(\frac{1}{\eta\delta})$ iterations} (2.65, .45) .. controls (2.79, .80) .. (2.8,.98);
        \draw[decorate,line width=.5mm,green!50!black,->] (2.8,1) to[out=100,in=0] node[pos=.5,above=1mm]{\small Stage III} (2.21,1.55);

        \path[red!90!black] (1.4,.349) node[inner sep=1mm,cross]{} node[right,yshift=2.5mm,fill=white,inner ysep=.0mm,inner xsep=1mm] {$t=T_s$};
        \path[blue] (2.755,.7) node[inner sep=1mm,cross]{} node[left=-.5mm,yshift=3mm] {$t=T_h$};
        \path[red!90!black] (2,.3) node[star,star point ratio=2.5,inner sep=.4mm,fill,draw=black,thick]{} node[xshift=0mm,below] {$t=T_1$};
        \fill[blue,draw=black,thick] (2.8,1) circle(.095cm) node[xshift=5mm,yshift=2.5mm] {$t=T_2$};
    \end{tikzpicture}\\
    \caption{Pictorial depiction of the three stages incurred by the OFTRL dynamics in the game $A_\delta$ defined in \eqref{eq:A delta}. The point $z^*$ denotes the unique Nash equilibrium. The times $T_1$ and $T_2$ are shown for concrete instantiations of OFTRL in \Cref{fig:intro plots} by a red star and a blue dot, respectively. The times $T_s$ and $T_h$ are defined in the proof of \Cref{theorem: main} in \Cref{app:proof main}.}
    \label{fig:three stages}
\end{figure}


\paragraph{Proof Sketch:} We decompose the analysis into three stages as illustrated in \Cref{fig:three stages}. We describe the three stages and the high-level ideas of our proof below and defer the full proof to \Cref{app:proof main}.
\begin{itemize}[leftmargin=*]
    \item \textbf{Stage I $[1, T_1 -1]$:} Recall that $x^1[1]=y^1[1] = \frac{1}{2}$ by \Cref{assumption:standard}. We show that $x^t[1]$ increases and denote $T_1$ the first iteration that $x^{T_1}[1] \ge \frac{1}{1+\delta}$. During the time $[1, T_1-1]$, since $x^t[1]$ is always smaller than $ \frac{1}{1+\delta}$, we know from \Cref{proposition:properties of A} action $1$ has larger loss than action $2$ for the $y$-player, i.e., $e^t_y = \ell^t_y[1] - \ell^t_y[2] \ge 0$. Thus $y^t[1]$ decreases during stage I and we show that $y^{T_1}[1] \le \frac{1}{2}-c_1$ with $c_1$ defined in \Cref{assumption:main}.
    \item \textbf{Stage II $[T_1, T_2]$:} Recall that $y^{T_1}[1] \le \frac{1}{2}-c_1$. We define $T_2 > T_1$ as the first iteration where $y^{T_2}[1] \ge \frac{1}{2(1+\delta)} > \frac{1}{2} - c_1$. 
    We remark that for $y^t[1]$ to increase, the loss vector must satisfy $e^t_y < 0$. However, the game matrix $A_\delta$ guarantees that $e^t_y \ge -\delta$ no matter what the $x$-player's strategy is (\Cref{proposition:properties of A}). Thus by the $\eta L$-Lipschitzness of $F_{\eta, R}$ (\Cref{prop: lip}), the per-iteration increase in $y^t[1]$ is at most $\eta L \delta$. Therefore, we know $T_2 - T_1 = \Omega(\frac{c_1}{\eta L \delta})$. As a result, $e^t_x < 0$ during $[T_1, T_2]$ and the cumulative loss for the $x$-player decreases to $E^{T_2}_x \le E^{T_1}_x - \Omega(\frac{1}{\eta L \delta})$. Recall $x^{T_1}[1] \ge \frac{1}{1+\delta}$. Thus $x^{T_2}[1] > x^{T_1}[1]$ is much closer to $1$.
    \item \textbf{Stage III $[T_2, T_3]$:} Recall that $y^{T_2}[1] \ge \frac{1}{2(1+\delta)}$.  Moreover, $y^t[1]$ could keep increasing if $x^t[1] \ge \frac{1}{1+\delta}$ since that implies $e^t_y \le 0$. Now the question is how long would the $x$-player stay close to the boundary, i.e, $x^t[1] \ge \frac{1}{1+\delta}$. Since OFTRL-type algorithms are not forgetful, this happens only when $E^{t}_x \ge E^{T_1}_x$ (recall $x^{T_1}[1] \ge \frac{1}{1+\delta}$). But we have at the end of stage II, $E^{T_2}_x \le E^{T_1}_x - \Omega(\frac{1}{\eta L \delta})$. Since the per-iteration loss is bounded by $1$, it requires at least $\Omega(\frac{1}{\eta L \delta})$ iterations to cancel the cumulative loss of $\Omega(\frac{1}{\eta L \delta})$. Define $T_3 = T_2 + \Omega(\frac{1}{\eta L \delta})$. During $[T_2, T_3]$, the $y$-player always receives loss such that $e^t_y \le 0$ and we prove that in the end $y^{T_3}[1] \ge \frac{1}{2} + c_2$ for some constant $c_2$.
    \item \textbf{Conclusion:} Finally, we get one iteration $T_3 \ge \Omega(\frac{1}{\eta L \delta})$ with $x^{T_3}[1] \ge \frac{1}{1+\delta}$ and $y^{T_3}[1] \ge \frac{1}{2} + c_2$. Using \Cref{lem:duality gap is large}, the duality gap of $(x^{T_3}, y^{T_3})$ is at least $c_2 > 0$.
\end{itemize}

\Cref{theorem: main} immediately implies the following (proof deferred to \Cref{app: last-iterate}). 
\begin{theorem}
    \label{theorem: no last-iterate rate}
    For optimistic FTRL with any regularizer satisfying \Cref{assumption:standard} and \Cref{assumption:main} and constant steps size $\eta \le \frac{1}{4L}$ ($L$ is defined in \Cref{assumption:standard}), there is no function $f$ such that
    the corresponding learning dynamics $\{(x^t,y^t)\}_{t\geq 1}$ in two-player zero-sum games $[0,1]^{d_1 \times d_2}$ has a last-iterate convergence rate of $f(d_1,d_2, T)$. More specifically, no function $f$ can satisfy
    \begin{itemize}
        \item[1.] $\DualGap(x^T, y^T) \le f(d_1,d_2, T)$ for all $[0,1]^{d_1 \times d_2}$ and for all $T \ge 1$.
        \item[2.] $\lim_{T\rightarrow \infty}f(d_1,d_2, T) \rightarrow 0$.
    \end{itemize}
\end{theorem}
\Cref{theorem: main} and \Cref{theorem: no last-iterate rate} provide impossibility results for getting a last-iterate convergence rate for OFTRL that solely depends on the bounded parameters, even in two-player zero-sum games. Moreover, they show the necessity of forgetfulness for fast last-iterate convergence in games since OGDA has a last-iterate convergence rate of $O(\frac{\poly(d_1, d_2)}{\sqrt{T}})$~\citep{cai2022finite,gorbunov2022last}.

\section{Extension to Higher Dimensions}

In this section, we extend our negative results from $2\times2$ matrix games to games with higher dimensions.
We start by showing an equivalence result for a single player (say, the first player). We assume that a decision maker is using OFTRL with a 1-strongly convex (w.r.t. the $\ell_2$ norm) and separable regularizer $R(x) = R_1(x_1) + R_2(x_2)$ to choose decisions.
At a given time time $t$, they see a loss $\ell^t\in [0,1]^2$.

Now consider the following $2n$-dimensional decision problem:  The player uses OFTRL using the regularizer $\hat R(\hat x) = \sum_{i=1}^{n} R_1(\hat x_i) + \sum_{i=n+1}^{2n} R_2(\hat x_i)$, \textit{i.e.}, they use $R_1$ on the first half of actions, and $R_2$ on the second half. This is again a 1-strongly convex regularizer (w.r.t. the $\ell_2$ norm).
Suppose the decision maker sees the rescaled and \emph{duplicated} version of the losses $\ell^1,\ldots,\ell^T$ from the 2-dimensional case: $\hat \ell^t_i = \frac{1}{n^\alpha}\ell^t_1$ if $i \leq n$, and $\hat \ell^t_i = \frac{1}{n^\alpha} \ell^t_2$ if $i > n$. The parameter $\alpha$ will be chosen later based on the regularizer.

Now we wish to show that by choosing $\alpha$ in the right way, we get that the decisions for the $2$-dimensional and $2n$-dimensional OFTRL algorithms are equivalent.
Let $x^1,\ldots, x^T$ be the 2-dimensional OFTRL decisions, and let $\hat x^1,\ldots,\hat x^T$ be the $2n$-dimensional OFTRL decisions.
Then, we want to show that $\sum_{i=1}^{n} \hat x_i^t = x^t[1]$ and $\sum_{i=n+1}^{2n} \hat x_i^t = x^t[2]$ for all $t$.

\begin{lemma}
    \label{lem:symmetry of duplication}
    Let the losses $\hat \ell^1, \ldots, \hat \ell^T$ satisfy the duplication procedure given in the preceding paragraph.
    Then for any time $t$, we have $\hat x^t_1=\cdots = \hat x^t_{n}$ and $\hat x^t_{n+1} = \cdots = \hat x^t_{2n}$.
\end{lemma}
\begin{proof}
    Suppose not and let $\hat x^t$ be the corresponding solution. Then the optimal solution is such that $\hat x^t_i \ne \hat x^t_k$ for some $i,k$ both less than $n$, or both greater than $n$. But then, by symmetry, we have that there is more than one optimal solution to the OFTRL optimization problem at time $t$: the objective is exactly the same if we create a new solution where we swap the values of $\hat x^t_i$ and $\hat x^t_k$. This is a contradiction due to strong convexity.
\end{proof}

From \cref{lem:symmetry of duplication}, we have that the OFTRL decision problem in $2n$ dimensions can equivalently be written as a $2$-dimensional decision problem:
Since the first $n$ entries must be the same, we can simply optimize over that one shared value, say $x^t[1]$, which we use for all $n$ entries, and similarly we use $x^t[2]$ for the second half of the entries.
Let $\mathrm{Dupl}: \Delta^2 \rightarrow \Delta^{2n}$ be a function that maps the two-dimensional solution into the corresponding duplicated $2n$-dimensional solution.
The equivalent $2$-dimensional problem is then:
\begin{align*}
    \hat x^t & = \mathrm{Dupl}\left[\argmin_{x\in \frac{1}{n} \cdot \Delta^{2}} \left\{ \frac{n}{n^\alpha}\left\langle x, \sum_{\tau = 1}^{t-1} \ell^\tau + \ell^{t-1} \right\rangle + \frac{n}{\eta} R_1(x[1]) + \frac{n}{\eta}R_2(x[2]) \right\} \right] \\
             & = \mathrm{Dupl}\left[\frac{1}{n}\cdot \argmin_{x\in \Delta^{2}} \left\{ \frac{n}{n^\alpha}\left\langle \frac{1}{n} x, \sum_{\tau = 1}^{t-1} \ell^\tau + \ell^{t-1} \right\rangle + \frac{n}{\eta} R(x/n) \right\} \right]                   \\
             & = \mathrm{Dupl}\left[\frac{1}{n}\cdot \argmin_{x\in \Delta^{2}} \left\{ \left\langle x, \sum_{\tau = 1}^{t-1} \ell^\tau + \ell^{t-1} \right\rangle + \frac{n^{\alpha+1}}{\eta} R(x/n)  \right\} \right].
\end{align*}

The next theorem shows that we can choose $\alpha$ for different regularizers and construct $2n\times 2n$ loss matrices whose learning dynamics are equivalent to the learning dynamics in $2\times 2$ games given in the preceding sections. We defer the proof to \Cref{app:high-dimension}.
\begin{theorem}
    \label{theorem:reduction from n to 2}
    For any loss matrix $A\in [0,1]^{2\times 2}$, there exists a loss matrix $\hat A\in [0,n^{-\alpha}]^{2n\times 2n}$ such that for the Euclidean ($\alpha = 1$), entropy ($\alpha = 0$), Tsallis ($\beta \in (0,1)$ and $\alpha = -1+\beta$),
    and log ($\alpha = -1$) regularizers, the resulting OFTRL learning dynamics are equivalent in the two games.
\end{theorem}

Combining \Cref{theorem: main} and \Cref{theorem:reduction from n to 2}, we have the following:
\begin{corollary}
    In the same setup as \Cref{theorem:reduction from n to 2}, under \Cref{assumption:standard} and \Cref{assumption:main}, there exists a game matrix $\hat A_\delta \in [0,n^{-\alpha}]^{2n\times 2n}$ such that the OFTRL learning dynamics  with any step size $\eta \le \frac{1}{4L}$ satisfies the following: there exists an iteration $t \ge \frac{c_1}{3\eta L \delta}$ with a duality gap at least $c_2 n^{-\alpha}$.
\end{corollary}

Since $\alpha = 0$ for the entropy regularizer, the same results hold more generally for games where one player has more actions than the other. In particular, we can create a $2n\times 2m$ game such that the resulting dynamics are equivalent to those in a $2\times 2$ game. This does not work for the Euclidean and log regularizers because the rescaling factors would be different for the row and column players.

\section{Conclusion and Discussions}
In this paper, we study last-iterate convergence rates of OFTRL algorithms with various popular regularizers, including the popular OMWU algorithm. Our main results show that even in simple $2 \times 2$ two-player zero-sum games parametrized by $\delta>0$, the lack of forgetfulness of OFTRL leads to the duality gap remaining constant even after $1/\delta$ iterations (\Cref{theorem: main}). As a corollary, we show that the last-iterate convergence rate of OFTRL must depend on a problem-dependent constant that can be arbitrarily bad (\Cref{theorem: no last-iterate rate}). This highlights a stark contrast with OOMD algorithms: while OGDA with constant step size achieves a $O(\frac{1}{\sqrt{T}})$ last-iterate convergence rate, such a guarantee is impossible for OMWU or more generally OFTRL. We now discuss several interesting questions regarding the convergence guarantees of learning in games and leave them as future directions.
\paragraph{Best-Iterate Convergence Rates} While we focus on the last-iterate (\textit{i.e.}, $\DualGap(x^T, y^T)$), the weaker notion of best-iterate (\textit{i.e.}, $\min_{t \in [T]} \DualGap(x^t, y^t)$) is also of both practical and theoretical interest. By definition, we know the best-iterate convergence rate is at least as good as the last-iterate convergence rate and could be much faster. This raises the following question:
\begin{align*}
    \textit{What is the best-iterate convergence rate of OMWU/OFTRL?}
\end{align*}
To our knowledge, there are no concrete results on the best-iterate convergence rates of OMWU or other OFTRL algorithms. 
It is thus interesting to extend our negative results to the best-iterate convergence rates 
or develop fast best-iterate convergence rates of OMWU/OFTRL.

\paragraph{Dynamic Step Sizes} Our negative results hold for OFTRL with \emph{fixed} step sizes. We conjecture that the slow last-iterate convergence of OFTRL persists even with \emph{dynamic} step sizes. In particular, we believe our counterexamples still work for OFTRL with decreasing step sizes. This is because decreasing the step size makes the players move even slower, and they may be trapped in the wrong direction for a longer time due to the lack of forgetfulness.  In \Cref{app:simu adaptive}, we include numerical results for OMWU with adaptive stepsize akin to Adagrad~\citep{duchi2011adaptive}, which supports our intuition. We observe the same cycling behavior as for fixed step size. While the cycle is smaller than that of fixed step sizes, the dynamics take more steps to finish each cycle. Investigating the effect of dynamic step sizes on last-iterate convergence rates is an interesting future direction.

\paragraph{Slow Convergence due to Lack of Forgetfulness} Our work shows that various OFTRL-type algorithms do not have fast last-iterate convergence rates for learning in games. Our proof and hard game instance build on the intuition that these algorithms lack forgetfulness: they do not forget the past quickly. This intuition is also utilized in~\citep{panageas2023exponential}. In particular, they give an $d \times d$  potential game where the last-iterate convergence rate of the Fictitious Play algorithm, which is equivalent to the Follow-the-Leader (FTL) algorithm, suffers exponential dependence in the dimension $d$. One natural future direction is to formalize the intuition of non-forgetfulness further and give a general condition for algorithms under which they suffer slow last-iterate convergence. It is also interesting to show other lower-bound results for learning in games.

\newpage
\subsubsection*{Acknowledgements}
We thank the anonymous reviewers for their constructive comments on improving the paper. Yang Cai was supported by the NSF Awards CCF-1942583 (CAREER) and CCF-2342642.
Christian Kroer was supported by the Office of Naval Research awards N00014-22-1-2530 and N00014-23-1-2374, and the National Science Foundation awards IIS-2147361 and IIS-2238960. Julien Grand-Cl{\'e}ment was supported by Hi! Paris and
Agence Nationale de la Recherche (Grant 11-LABX-0047).
Haipeng Luo was supported by NSF award IIS-1943607. Weiqiang Zheng was supported by the NSF Awards CCF-1942583 (CAREER), CCF-2342642, and a Research Fellowship
from the Center for Algorithms, Data, and Market Design at Yale (CADMY).
\bibliographystyle{plainnat}
\bibliography{ref}

\appendix
\tableofcontents

\section{Missing Proofs in \Cref{sec: pre}}
\label{app:pre}
\subsection{Proof of \Cref{lemma: F monotone}}
\begin{proof}
    Let $e_1 < e_2$. Denote $x_1 = F_{\eta, R}(e_1)$ and $x_2 = F_{\eta, R}(e_2)$. By definition, we have
    \begin{align*}
        e_2(x_2- x_1) \le \frac{1}{\eta}\InParentheses{R(x_1) - R(x_2)} \le e_1(x_2 - x_1).
    \end{align*}
    Since $e_1 < e_2$, we have $x_2 \le x_1$.
\end{proof}

\subsection{Proof of \Cref{prop: lip}}
\begin{proof}
    By definition,
    \begin{align*}
        F_{\eta, R}\InParentheses{\frac{E}{\eta}} & = \argmin_{x \in [0,1]}\left\{ x \cdot \frac{E}{\eta} + \frac{1}{\eta} R(x)
        \right\}
        =  \argmin_{x \in [0,1]}\left\{ x \cdot E + R(x) \right\}
        = F_{1, R}(E).
    \end{align*}
    The second claim on the Lipschitzness follows directly.
\end{proof}

\section{Missing Proofs in \Cref{sec:analysis}}
\label{app:analysis}

\subsection{Proof of \Cref{lem:duality gap is large}}
\label{app:duality gap}
\begin{proof}
    We have
    \begin{align*}
        \DualGap(x, y) & = \max_{\tilde{y} \in \Delta^2}  x ^\top A_{\delta} \tilde{y} - \min_{\tilde{x} \in \Delta^2} \tilde{x}^\top A_{\delta} y \\
                       & = \max_{i\in \{1,2\}} (A_{\delta}^\top x)[i] - \min_{i\in \{1,2\}} (Ay_{\delta})[i]                                       \\
                       & = (\frac{1}{2} + \delta) x[1] - (1 - y[1]) \tag{$x[1] \geq \frac{1}{1+\delta},\epsilon>0$}                                \\
                       & \ge \frac{1}{2} \frac{1 + 2\delta}{1 + \delta} - \frac{1}{2}+\epsilon                                                     \\
                       & \ge \epsilon.
    \end{align*}
\end{proof}

\subsection{Proof of \Cref{theorem: main}}
\label{app:proof main}
\begin{proof}
    Recall that $c_1 = \frac{1}{2} - F_{1, R}(\frac{1}{20L})$ defined in \Cref{assumption:main}. We fix any $\delta < \min\{\frac{1}{15}, \frac{c_1}{6}, \frac{c_1^2}{300}, \delta'\}$.
    Since $\delta < \delta'$, \Cref{assumption:main} holds.
    We will prove that there exists an iteration $t \ge \frac{c_1}{3\eta L\delta}$ with duality gap $c_2$.

    \paragraph{Proof Plan:} We decompose the analysis into three stages as shown in \Cref{fig:three stages}. Below, we describe the three stages and the high-level ideas in our proof.
    \begin{itemize}[leftmargin=*]
        \item \textbf{Stage I $[1, T_1 -1]$:} Recall that $x^1[1]=y^1[1] = \frac{1}{2}$ by \Cref{assumption:standard}. We show that $x^t[1]$ increases and denote $T_1$ the first iteration that $x^{T_1}[1] \ge \frac{1}{1+\delta}$. During the time $[1, T_1-1]$, since $x^t[1]$ is always smaller than $ \frac{1}{1+\delta}$, we know from \Cref{proposition:properties of A} action $1$ has larger loss than action $2$ for the $y$-player, i.e., $e^t_y = \ell^t_y[1] - \ell^t_y[2] \ge 0$. Thus $y^t[1]$ decreases during stage I and we show that $y^{T_1}[1] \le \frac{1}{2}-c_1$ with $c_1$ defined in \Cref{assumption:main}.
        \item \textbf{Stage II $[T_1, T_2]$:} Recall that $y^{T_1}[1] \le \frac{1}{2}-c_1$. We define $T_2 > T_1$ as the first iteration where $y^{T_2}[1] \ge \frac{1}{2(1+\delta)} > \frac{1}{2} - c_1$. 
        We remark that for $y^t[1]$ to increase, the loss vector must satisfy $e^t_y < 0$. However, the game matrix $A_\delta$ guarantees that $e^t_y \ge -\delta$ no matter what the $x$-player's strategy is (\Cref{proposition:properties of A}). Thus by the $\eta L$-Lipschitzness of $F_{\eta, R}$ (\Cref{prop: lip}), the per-iteration increase in $y^t[1]$ is at most $\eta L \delta$. Therefore, we know $T_2 - T_1 = \Omega(\frac{c_1}{\eta L \delta})$. As a result, $e^t_x < 0$ during $[T_1, T_2]$ and the cumulative loss for the $x$-player decreases to $E^{T_2}_x \le E^{T_1}_x - \Omega(\frac{1}{\eta L \delta})$. Recall $x^{T_1}[1] \ge \frac{1}{1+\delta}$. Thus $x^{T_2}[1] > x^{T_1}[1]$ is much closer to $1$.
        \item \textbf{Stage III $[T_2, T_3]$:} Recall that $y^{T_2}[1] \ge \frac{1}{2(1+\delta)}$.  Moreover, $y^t[1]$ could keep increasing if $x^t[1] \ge \frac{1}{1+\delta}$ since that implies $e^t_y \le 0$. Now the question is how long would the $x$-player stay close to the boundary, i.e, $x^t[1] \ge \frac{1}{1+\delta}$. Since OFTRL-type algorithms are not forgetful, this happens only when $E^{t}_x \ge E^{T_1}_x$ (recall $x^{T_1}[1] \ge \frac{1}{1+\delta}$). But we have at the end of stage II, $E^{T_2}_x \le E^{T_1}_x - \Omega(\frac{1}{\eta L \delta})$. Since the per-iteration loss is bounded by $1$, it requires at least $\Omega(\frac{1}{\eta L \delta})$ iterations to cancel the cumulative loss of $\Omega(\frac{1}{\eta L \delta})$. Define $T_3 = T_2 + \Omega(\frac{1}{\eta L \delta})$. During $[T_2, T_3]$, the $y$-player always receives loss such that $e^t_y \le 0$ and we prove that in the end $y^{T_3}[1] \ge \frac{1}{2} + c_2$ for some constant $c_2$.
        \item \textbf{Conclusion:} Finally, we get one iteration $T_3 \ge \Omega(\frac{1}{\eta L \delta})$ with $x^{T_3}[1] \ge \frac{1}{1+\delta}$ and $y^{T_3}[1] \ge \frac{1}{2} + c_2$. Using \Cref{lem:duality gap is large}, the duality gap of $(x^{T_3}, y^{T_3})$ is at least $c_2 > 0$.
    \end{itemize}

    \paragraph{Stage I:} OMWU is initialized with $x^1[1] = y^1[1] = \frac{1}{2}$. We define two time steps for Stage I's analysis:
    \begin{itemize}
        \item $T_s > 1$: the smallest iteration such that $x^{T_s}[1] \ge \frac{3}{4}$
        \item $T_1 > T_s$: the smallest iteration such that $x^{T_1}[1] \ge \frac{1}{1+\delta}$.
    \end{itemize}
    We remark that $T_s$ and $T_1$ exist and postpone the proof of this fact in \Cref{claim: T_s T_1 existence} at the end of this paragraph.

    Notice that by Proposition \ref{proposition:properties of A}, the difference $e^t_x = \ell^t_x[1] - \ell^t_x[2]$ is lower bounded: $e^t_x \ge -\frac{1}{2}$ for any $t \ge 1$. Thus $E^{t-1}_x + e^{t-1}_x \ge -\frac{t}{2}$ for any $t \ge 1$. Since $x^{T_s}[1] \ge \frac{3}{4} > \frac{1}{2}$, we know that $E^{T_s-1}_x + e^{T_s-1}_x < 0$. Then, by $\eta L$-Lipschitzness of $F_{\eta, R}$, we have
    \[
        \frac{1}{4} \le x^{T_s}[1] - x^1[1] \le \eta L \cdot \left| E^{T_s-1}_x + e^{T_s-1}_x \right| \le \frac{L\eta T_s}{2}.
    \]
    This implies
    \[
        T_s \ge \frac{1}{2\eta L}.
    \]
    Since $x^t[1] < \frac{3}{4}$ for all $1 \le t \le T_s-1$, we know that $e^t_y = \ell^t_y[1] - \ell^t_y[2] =  1 - (1+\delta)x^t[1] > \frac{1-3\delta}{4} \ge \frac{1}{5}$ (as $\delta \le \frac{1}{15}$) for all $1 \le t \le T_s-1$. Moreover, for all $1 \le t \le T_1-1$, we know that $e^t_y \ge 0$ as $x^t[1] \le \frac{1}{1+\delta}$. Since the difference $e^t_y$ is at least $\frac{1}{5}$ for all $t \in [1, T_s-1]$ and remains non-negative for all $t\in[1,T_1-1]$, we have that for all $\frac{1}{2\eta L} \le t \le T_1$
    \begin{align*}
        y^t[1] &= F_{1, R}(E^{t-1}_y + e^{t-1}_y) \\
        & \le F_{\eta, R}(E^{t-1}_y) \tag{$e^{t-1}_y\ge 0$}  \\
        & \le F_{\eta, R}\InParentheses{\frac{T_s-1} {5}}   \tag{$t \ge T_s$ and $e^t_y \ge \frac{1}{5}$ for all $t \in [1, T_s-1]$}                                                               \\
                               & \le F_{\eta, R}\InParentheses{\frac{1}{20L\eta}} \tag{$T_s - 1 \ge \frac{1}{2\eta L} - 1 \ge \frac{1}{4\eta L}$} \\
                               & = F_{1,R}\InParentheses{\frac{1}{20L}} =\frac{1}{2} - c_1.
    \end{align*}
    This completes the proof of Stage I, where $x^{T_1}[1] \ge \frac{1}{1+\delta}$ and $y^{T_1}[1] \le \frac{1}{2} - c_1$. Before we proceed to the next stage, we prove the existence of $T_s$ and $T_1$.

    \begin{claim}
        \label{claim: T_s T_1 existence}
        $T_s$ and $T_1$ exist.
    \end{claim}
    \begin{proof}
        It suffices to prove that $T_1$ exists as it implies the existence of $T_s$. Assume for the sake of contradiction that $T_1$ does not exist, \textit{i.e.}, $x^t[1] < \frac{1}{1+\delta}$ for all $t \ge 1$. By the same analysis as for Stage I, we get $y^t[1] \le \frac{1}{2} - c_1$ for all $t \ge \frac{1}{2\eta L}$. This implies $e^t_x = -\frac{1}{2} + (1+\delta) y^t[1] \le \frac{\delta}{2} - c_1 \le -\frac{c_1}{2}$ for all $t \ge \frac{1}{2\eta L}$. Then $E^t_x + e^t_x \rightarrow -\infty$ as $t \rightarrow +\infty$. As a consequence, $x^t[1] = F_{\eta, R}(E^{t-1}_x + e^{t-1}_x) \rightarrow 1$ as $t \rightarrow +\infty$ by item 2 in \Cref{assumption:standard}. But this contradicts with the assumption that $x^t[1] < \frac{1}{1+\delta}$ for all $t \ge 1$. This completes the proof. 
    \end{proof}

    \paragraph{Stage II} We define
    \begin{align}\label{eq: dfn T_h}
        T_h := \left\lfloor \frac{c_1}{2L\eta \delta} \right\rfloor \in \InBrackets{\frac{c_1}{3L\eta \delta}, \frac{c_1}{2L\eta \delta}},
    \end{align}
    where the lower bound on $T_h$ holds since $ \frac{c_1}{6L\eta\delta}\ge \frac{c_1}{6\delta} \ge 1$. 

    In Stage I, we have proved that $y^{T_1}[1] \le \frac{1}{2} - c_1$.   We claim that for all $t \in [T_1, T_1 + T_h- 1]$, $y^t[1] \le \frac{1}{2} - \frac{c_1}{2}$. To prove the claim, we first notice that $-\delta \le e^t_y \le 1$ for all $t \ge 1$. Then by the monotonicity and the $\eta L$-Lipschitzness of $F_{\eta, R}$ (\Cref{lemma: F monotone} and \Cref{lemma:Lipschitz}), we get for all $t \in [T_1, T_1 + T_h- 1]$,
    \begin{align*}
        y^{t}[1] & \le F_{\eta, R}(E^{T_1-1}_y) + \eta L \max\left\{E^{T_1-1}_y - E_y^{t-1} - e^{t-1}_y, 0\right\} \\
                 & \le \frac{1}{2} - c_1 + \eta L \cdot (t-T_1+1) \delta                                           \\
                 & \le \frac{1}{2} - c_1 + \eta L T_h \delta
        \\
                 & \le \frac{1}{2} - \frac{c_1}{2},
    \end{align*}
    where, in the second-to-last inequality, we use $t - T_1 + 1 \le T_h \le \frac{c_1}{2\eta L \delta}$ by \Cref{eq: dfn T_h}.

    Now we denote $T_2 \ge T_1$ the smallest iteration when $y^{T_2}[1] \ge \frac{1}{2(1+\delta)}$. The existence of $T_2$ will become clear in the following analysis, and we postpone the proof to \Cref{claim: T_2 exists} at the end of the discussion. Then for all $t \in [T_s, T_2-1]$, we have $y^t[1] \le \frac{1}{2(1+\delta)}$, which implies $e^t_x \le 0$. Moreover, for all $t \in [T_s, T_1 + T_h - 1]$, since $y^t[1] \le \frac{1}{2} - \frac{c_1}{2}$, we have
    \begin{align*}
        e^t_x = \ell^t_x[1] - \ell^t_x[2] & = -\frac{1}{2} + (1+\delta) y^t[1]     \le \frac{-1 + (1+\delta)(1-c_1)}{2} \le \frac{\delta - c_1}{2}          \le - \frac{c_1}{4},
    \end{align*}
    where in the last inequality we use $\delta \le \frac{c_1}{2}$.
    Then for any $T_1 + T_h \le t \le T_2$, we have
    \begin{align}
         x^t[1] = F_{\eta, R}(E^{t-1}_x + e^{t-1}_x)                                       
               & \ge F_{\eta, R}(E_x^{T_1 + T_h - 1})    \tag{$e^{t-1}_x \le 0$ for all $t \in [T_1 + T_h, T_2]$} \\
               & \ge F_{\eta, R}\InParentheses{-\frac{c_1 T_h}{4} + E_x^{T_1-1}}                               \nonumber \\
               & \ge F_{\eta, R}\InParentheses{-\frac{c_1 T_h}{5} + E_x^{T_1-1} + e^{T_1-1}_x}, \label{eq: x T_2 bound}
    \end{align}
    where in the second last inequality, we use the fact that $\frac{c_1 T_h}{20} \ge \frac{c_1^2}{60 \eta L \delta} \ge 1$.

    \begin{claim}\label{claim: T_2 exists}
        $T_2$ exists.
    \end{claim}
    \begin{proof}
        Assume for the sake of contradiction that $T_2$ does not exist, \textit{i.e.}, $y^t[1] < \frac{1}{2(1+\delta)}$ for all $t \ge T_1$ (since we know $y^t[1] \le \frac{1}{2} - \frac{c_1}{2}$ for all $t \in [T_1, T_1 + T_h -1]$). Then by the analysis of Stage II and \Cref{eq:x condition}, we have $x^t[1] \ge \frac{1+c_3}{1+c_3+\delta}$ for all $t \ge T_1$. This implies $e^t_y \le -\frac{c_3\delta}{2}$ for all $t \ge T_1$. As a result, we have $E^{t-1}_y + e^{t-1}_y \rightarrow - \infty$ as $t \rightarrow \infty$. By item 2 in \Cref{assumption:standard}, we get $y^t[1] = F_{\eta, R}(E^{t-1}_y + e^{t-1}_y) \ge \frac{1}{2}$ as $t \rightarrow \infty$. But this contradicts with the assumption that $y^t[1] < \frac{1}{2(1+\delta)}$ for all $t \ge T_1$. This completes the proof.
    \end{proof}

    \paragraph{Stage III}
    Recall that we have argued in State I that $F_{\eta, R}(E_x^{T_1-1} + e^{T_1-1}_x) = F_{1, R}(\eta (E_x^{T_1-1} + e^{T_1-1}_x)) = x^{T_1}[1] \ge \frac{1}{1+\delta}$. We thus have
    \begin{align}
        \label{eq:x condition}
        F_{\eta, R}\InParentheses{-\frac{c_1 T_h}{10} + E_x^{T_1-1} + e^{T_1-1}_x)}
         & \ge F_{\eta, R}\InParentheses{-\frac{c_1^2}{30 L \eta \delta} + E_x^{T_1-1} + e^{T_1-1}_x)}  \nonumber \\
         & = F_{1, R}\InParentheses{-\frac{c_1^2}{30 L \delta} + \eta (E_x^{T_1-1} + e^{T_1-1}_x))}  \nonumber    \\
         & \ge \frac{1+c_3}{1+ c_3 + \delta},
    \end{align}
    where the first inequality follows $T_h \ge \frac{c_1}{3\eta L \delta}$ and the monotonicity of $F_{\eta,R}$ (\Cref{lemma: F monotone}); the last inequality is by item 1 in \Cref{assumption:main}.

    Now denote $T_3 = T_2 + \lfloor \frac{c_1 T_h}{10} \rfloor - 2$. For any $T_2 \le t \le T_3$, we know that
    \begin{align*}
        x^t[1] & =  F_{\eta, R}(E^{t-1}_x + e^{t-1}_x)                                                                                 \\
               & =  F_{\eta, R}(E^{T_2-1}_x + e^{T_2-1}_x + \sum_{k=T_2}^{t-1} e^{k}_x + e^{t-1}_x - e^{T_2-1}_x)                      \\
               & \ge F_{\eta, R}(- \frac{c_1 T_h}{5} + E_x^{T_1-1} + e^{T_1-1}_x + \sum_{k=T_2}^{t-1} e^{k}_x + e^{t-1}_x - e^{T_2-1}_x) \tag{by \eqref{eq: x T_2 bound}}\\
               & \ge  F_{\eta, R}(- \frac{c_1 T_h}{5} + E_x^{T_1-1} + e^{T_1-1}_x + \frac{c_1 T_h}{10} - 2 + 2)                          \tag{$t- T_2\le \frac{c_1 T_h}{10}-2$}  \\
               & = F_{\eta, R}(-\frac{c_1 T_h}{10} + E_x^{T_1-1} + e^{T_1-1}_x))           \ge \frac{1 + c_3}{1 + c_3 + \delta}. \tag{by \eqref{eq:x condition}}
    \end{align*}
    Note that $1+ c_3 + \delta \le 2$. This implies $e^t_y = 1 - (1+\delta)x^t[1] = -\frac{c_3\delta}{1+c_3 + \delta} \le -\frac{c_3\delta}{2}$ for all $T_2 \le t \le T_3$. Moreover, we know that $e^t_y \ge - \delta$ for any $t$. Then for any $t \in [T_2 + \lceil \frac{c_1T_h}{20}\rceil, T_3]$
    \begin{align*}
        y^{t}[1] & = F_{\eta, R}( E^{t-1}_y + e^{t-1}_y))                                                                                                                     \\
                   & = F_{\eta, R}( E^{T_2-1}_y + e^{T_2-1}_y + \sum_{k= T_2}^{t-1}e^{k}_y + e^{t-1}_y - e^{T_2-1}_y))                                                        \\
                   & \ge F_{\eta, R}( E^{T_2-1}_y + e^{T_2-1}_y - \frac{c_3\delta (t - T_2)}{2} + \delta)                                                                         \\
                   & \ge F_{\eta, R}( E^{T_2-1}_y + e^{T_2-1}_y - \frac{c_3\delta c_1 T_h}{40} + \delta) \tag{$t- T_2  \ge \frac{c_1 T_h}{20}$} \\
                   & \ge F_{\eta, R}( E^{T_2-1}_y + e^{T_2-1}_y - \frac{c_3 c_1^2}{120 \eta L} + \delta) \tag{$T_h \ge \frac{c_1}{3 \eta L\delta}$}                                   \\
                   & = F_{1, R}(\eta (E^{T_2-1}_y + e^{T_2-1}_y) - \frac{c_3c_1^2}{120 L} + \eta \delta)                                                                            \\
                   & \ge  F_{1, R}(\eta (E^{T_2-1}_y + e^{T_2-1}_y) - \frac{c_3c_1^2}{120 L} + \frac{\delta}{4L}). \tag{$\eta \le \frac{1}{4L}$}
    \end{align*}
    Recall that $F_{1, R}(\eta (E^{T_2-1}_y + e^{T_2-1}_y))= F_{\eta, R}( E^{T_2-1}_y + e^{T_2-1}_y) =  y^{T_2}[1] \ge \frac{1}{2(1+\delta)}$. By item 2 in \Cref{assumption:main}, we have $F_{1, R}(\eta(E^{T_2-1}_y + e^{T_2-1}_y) - \frac{c_1^2}{120L} + \frac{\delta}{4L}) \ge \frac{1}{2} + c_2$ for some absolute constant $c_2 > 0$. Thus,  we have $y^{t}[1] \ge \frac{1}{2} + c_2$ for all $[T_2 + \lceil \frac{c_1T_h}{20}\rceil, T_3]$. Recall that $x^{t}[1] \ge \frac{1+c_3}{1+c_3 + \delta} \ge \frac{1}{1+\delta}$ for all $t \in [T_2, T_3]$. Then by \Cref{lem:duality gap is large} we can conclude that the duality gap of $(x^{t}, y^{t})$ is at least $c_2 > 0$ for all $t \in [T_2 + \lceil \frac{c_1T_h}{20}\rceil, T_3]$. This completes the proof as $T_2 \ge T_h \ge \frac{c_1}{3\eta L\delta}$.
\end{proof}

\subsection{Proof of \Cref{theorem: no last-iterate rate}}
\label{app: last-iterate}
\begin{proof}
    Assume for the sake of contradiction that there is a function that satisfies both conditions. Then for any $A \in [0,1]^{2 \times 2}$, we have the OFTRL learning dynamics over $A$ satisfies
    \begin{itemize}
        \item[1.] $\DualGap(x^T, y^T) \le f(2,2,T)$ for all $T$.
        \item[2.] $\lim_{T \rightarrow \infty} f(2,2,T) \rightarrow 0$.
    \end{itemize}
    Since $\lim_{T \rightarrow \infty} f(2,2,T) \rightarrow 0$, we know there exists $T_0 > 0$ such that for any $t \ge T_0$, $\DualGap(x^t, y^t) \le f(2,2,t) < c_2$. Now let $\delta \le \min\{\hat \delta, \frac{c_1}{3\eta L T_0}\}$. Then by \Cref{theorem: main}, we know there exists an iteration $t \ge \frac{c_1}{3\eta L \delta} \ge T_0$ such that $\DualGap(x^t, y^t) \ge c_2$. This contradiction completes the proof.
\end{proof}

\section{Verifying \Cref{assumption:main} for Different Regularizers}
\label{app:verification}

\begin{lemma}\label{lemma:Lipschitz}
    If the regularizer $R$ is $1$-strongly convex, then $F_{1, R}$ is $\frac{1}{2}$-Lipschitz.
\end{lemma}
\begin{proof}
    Notice that $R(x) + R(1-x)$ is $2$-strongly convex. Thus by standard analysis (see e.g.,~\citet[Lemma~4]{luo_lecture2}) we know $F_{1, R}$ is $\frac{1}{2}$-Lipschitz.
\end{proof}
By \Cref{lemma:Lipschitz}, we can choose $L = \frac{1}{2}$ for any $1$-strongly convex regularizer in \Cref{assumption:standard}.

\subsection{Negative Entropy}
\begin{lemma}[\Cref{assumption:main} holds for the entropy regularizer]
    Consider the negative entropy regularizer $R$ defined as $R(x) = x\log x + (1-x) \log (1-x)$.
    Then $F_{1,R}$ is $L = \frac{1}{2}$-Lipschitz. Moreover, $c_1=\frac{1}{2} - F_{1,R}(\frac{1}{20L})$ is a universal constant and \Cref{assumption:main} holds with $\delta' = \frac{c_1^2}{480L}$, $c_2 = F_{1,R}(-\frac{c_1^2}{480L})- \frac{1}{2}$, and $c_3 = \frac{1}{2}$.
\end{lemma}
\begin{proof}
    It is easy to verify that $F_{1, R}(x)$ has a closed-form representation
    \begin{align*}
        F_{1,R}(E) = \frac{1}{1+\exp(E)}.
    \end{align*}
    Thus $L = \frac{1}{2}$ and $c_1=\frac{1}{2} - F_{1,R}(\frac{1}{20L})$ is a universal constant.
    We also choose $c_3 = \frac{1}{2}$.

    If $F_{1, R}(E) \ge \frac{1}{1+\delta}$, then we have $E \le -\log(1/\delta)$. We note that
    \begin{align*}
        \exp\InParentheses{-\frac{c_1^2}{30L\delta}}\le \frac{1}{1+c_3} \Rightarrow \frac{1}{1 + \exp\InParentheses{-\frac{c_1^2}{30L\delta} - \log(1/\delta)}} \ge \frac{1 +c_3}{1+c_3 + \delta}.
    \end{align*}

    Thus $\delta \le \delta_1 = \frac{c_1^2}{30L\log(1+c_3))} = \frac{c_1^2}{30\log(\frac{3}{2}))L}$ suffices for item 1 in \Cref{assumption:main}.

    If $F_{1, R}(E) \ge \frac{1}{2(1+\delta)} = \frac{1}{1+1 + 2\delta}$, we have $E \le \log(1+2\delta)$. Note that since $\log(1+2y) \le 2y$ for $y > 0$, we have
    \begin{align*}
        \delta \le \frac{c_3c_1^2}{480L}
         & \Rightarrow - \frac{c_3c_1^2}{120L} + \log(1+2\delta) < -\frac{c_3c_1^2}{240L}                                  \\
         & \Rightarrow F_{1,R}\InParentheses{-\frac{c_3c_1^2}{120L} + E} > F_{1,R}\InParentheses{-\frac{c_3 c_1^2}{240L}}.
    \end{align*}
    Thus item 2 in \Cref{assumption:main} holds for any $\delta \le \delta_2 = \frac{c_3 c_1^2}{480L} = \frac{c_1^2}{960L}$ and $c_2 = F_{1,R}(-\frac{c_3c_1^2}{240L})- \frac{1}{2} = F_{1,R}(-\frac{c_1^2}{480L})- \frac{1}{2}$.

    Combining the above, we know \Cref{assumption:main} holds for the negative entropy regularizer with $\delta' = \frac{c_1^2}{960L}$ and $c_2 = F_{1,R}(-\frac{c_1^2}{480L})- \frac{1}{2}$.
\end{proof}

\subsection{Squared Euclidean Norm Regularizer}
\begin{lemma}[\Cref{assumption:main} holds for the Euclidean regularizer]
    Consider the Euclidean regularizer $R$ defined as $R(x) = \frac{1}{2}(x^2 + (1-x)^2)$. We have $L = \frac{1}{2}$ and $c_1 = \frac{1}{20}$. We also have \Cref{assumption:main} holds with $\delta' = \frac{c_1^2}{480L}$, $c_2 =\frac{c_1^2}{960L}$, and $c_3 = \frac{1}{2}$.
\end{lemma}
\begin{proof}
    It is easy to verify that $F_{1, R}(x)$ has a closed-form representation
    \begin{align*}
        F_{1,R}(x) = \begin{cases}
                      1               & \text{ if } x \le - 1       \\
                      \frac{1 - x}{2} & \text{ if }  x \in (-1,  1) \\
                      0               & \text{ if } x \ge 1
                  \end{cases}
    \end{align*}
    Thus $F_{1,R}$ is $L$-Lipschitz with $L = \frac{1}{2}$. Moreover, $c_1 = \frac{1}{2} - F_{1,R}(\frac{1}{20L}) = \frac{1}{20}$. We choose $c_3 = \frac{1}{2}$.

    Fix any $E$ such that $F_{1, R}(E) \ge \frac{1}{1+\delta}$. We have $E \le -\frac{1-\delta}{1+\delta} < 0$. We note that for any $\delta \le \frac{c_1^2}{30L} = \frac{c_1^2}{15}$,
    \begin{align*}
        F_{1,R}\InParentheses{-\frac{c_1^2}{30L\delta} + E} \ge F_{1,R}(-1) = 1.
    \end{align*}
    Thus $\delta \le \delta_1 = \frac{c_1^2}{30L}$ suffices for item 1 in \Cref{assumption:main}.

    Fix any $E$ such that $F_{1, R}(E) \ge \frac{1}{2(1+\delta)} = \frac{1}{2(1+\delta)}$. We have $E \le \frac{\delta}{1+\delta} \le \delta$. The for any $\delta \le \frac{c_3c_1^2}{240L}$, we have
    \begin{align*}
        F_{1,R}\InParentheses{-\frac{c_3c_1^2}{120L} + E} \ge F_{1,R}\InParentheses{-\frac{c_3c_1^2}{240L}} = \frac{1}{2} + \frac{c_3c_1^2}{480L}
    \end{align*}
    Thus item 2 in \Cref{assumption:main} holds for any $\delta \le \delta_2 = \frac{c_1^2}{480L}$ and $c_2 = \frac{c_1^2}{960L}$.

    Combining the above, we know \Cref{assumption:main} holds for the negative entropy regularizer with $\delta' = \min\{\delta_1, \delta_2\} = \frac{c_1^2}{480L}$ and $c_2 = \frac{c_1^2}{960L}$.
\end{proof}

\subsection{Log Barrier}

\begin{lemma}[\Cref{assumption:main} holds for the log barrier]
    Consider the log barrier regularizer $R$ defined as $R(x) = -\log(x) - \log(1-x)$.  Then \Cref{assumption:main} holds with the following choices of constants:
    \begin{itemize}
        \item[1.] $c_1 = \sqrt{\frac{1}{4} + 400L^2} - 20L > 0$.
        \item[2.] $c_3 = \frac{c_1^2}{60L}$.
        \item[3.] $c_2 = \sqrt{\frac{1}{4}+(\frac{c_3c_1^2}{240L})^2} -\frac{c_3c_1^2}{240L} > 0$.
        \item[4.] $\delta' = \frac{c_3c_1^2}{2160L}$.
    \end{itemize}
\end{lemma}
\begin{proof}
    By setting the gradient of $x \cdot E + R(x)$ to $0$, we get a closed-form expression of $F_{1,R}$:
    \begin{align*}
        F_{1,R}(E) =\begin{cases}
                        \frac{1}{2} + \frac{1}{E} -\sqrt{\frac{1}{4} + \frac{1}{E^2}}  & \text{ if } E > 0  \\
                        \frac{1}{2}                                                    & \text{ if } E = 0  \\
                        \frac{1}{2} + \frac{1}{E} + \sqrt{\frac{1}{4} + \frac{1}{E^2}} & \text{ if } E < 0.
                    \end{cases}
    \end{align*}
    For $x \in (0,1)$, the $F_{1,R}$ function admits an inverse function defined as
    \begin{align*}
        F^{-1}_{1,R}(x) = \frac{2x-1}{x^2-x}.
    \end{align*}
    Thus we know $E_0 := F^{-1}_{1,R}(\frac{1}{1+\delta}) = -\frac{1-\delta^2}{\delta}$ satisfies $F_{1,R}(E_0) = \frac{1}{1+\delta}$. Moreover, we can calculate
    \begin{align*}
        F^{-1}_{1,R}\InParentheses{\frac{1+c_3}{1+c_3+\delta}} & = -\frac{(1+c_3)^2 - \delta^2}{(1+c_3)\delta}         \\
                                                               & = - \frac{1+c_3}{\delta} + \frac{\delta}{1+c_3}       \\
                                                               & = E_0 - \frac{c_3}{\delta} - \frac{c_3\delta}{1+c_3}.
    \end{align*}
    Thus we can choose $c_3 = \frac{c_1^2}{60L}$ so that
    \begin{align*}
         & E_0 -\frac{c_1^2}{30L\delta}                                                                   \\
         & = E_0 -\frac{c_3}{\delta} - \frac{c_3}{\delta}                                                 \\
         & \le E_0 -\frac{c_3}{\delta} - \frac{c_3\delta}{1+c_3} \tag{since $\delta < 1/2$ and $c_3 > 0$}.
    \end{align*}
    Thus we have $F_{1,R}(E_0 - \frac{c_1^2}{30L\delta}) \ge F_{1,R}(E_0 -\frac{c_3}{\delta} - \frac{c_3\delta}{1+c_3}) \ge \frac{1+c_3}{1+c_3 + \delta}$.

    We calculate $E_1 := F^{-1}_{1,R}(\frac{1}{2(1+\delta)}) = \frac{4(\delta+\delta^2)}{1+2\delta} \le 8\delta$. Then we can choose $\delta \le \delta':= \frac{c_3c_1^2}{2160L}$. Then we have
    \begin{align*}
         & F_{1, R}(-\frac{c_3 c_1^2}{120L} + \frac{\delta}{4L} + E_1) \\
         & \ge  F_{1,R}(-\frac{c_3 c_1^2}{120L} + 9\delta)             \\
         & \ge F_{1,R}(-\frac{c_3 c_1^2}{240L})                        \\
         & = \frac{1}{2} + c_2,
    \end{align*}
    where $c_2 =  \sqrt{\frac{1}{4}+(\frac{c_3c_1^2}{240L})^2} -\frac{c_3c_1^2}{240L} > 0$ by the closed-form expression of $F_{1,R}$.
\end{proof}

\subsection{Negative Tsallis Entropy}
For $x \in [0,1]$, the negative Tsallis entropy is a family of regularizers parameterized by $\beta \in (0,1)$:
\begin{equation}
    R(x) = \frac{1 - x^\beta}{1-\beta}.
\end{equation}
The corresponding $F_{1,R}$ is defined as
\begin{align*}
    F_{1,R}(E) = \argmin_{x \in (0,1)} \left\{ x \cdot E + \frac{1 - x^\beta}{1-\beta} + \frac{1 - (1-x)^\beta}{1-\beta}  \right\}.
\end{align*}
For $x \in (0,1)$, we note that $F_{1,R}$ has an inverse function
\begin{align*}
    F^{-1}_{1,R}(x) = \frac{\beta}{1-\beta} \InParentheses{x^{\beta -1} - (1-x)^{\beta - 1}}.
\end{align*}

\begin{lemma}[\Cref{assumption:main} holds for Tsallis entropy]
    Consider Tsallis entropy parameterized by $\beta \in (0,1)$. Then $L = \frac{1}{2\beta}$ and \Cref{assumption:main} holds with the following choices of constants:
    \begin{itemize}
        \item[1.] $c_1 = \frac{1}{2} - F_{1,R}(\frac{1}{20L}) > 0$.
        \item[2.] $c_3 = \frac{1}{2}$.
        \item[3.] $c_2 = F_{1,R}(-\frac{c_3c_1^2}{240L}) - \frac{1}{2} > 0$.
        \item[4.] $\delta' = \min\{\InParentheses{\frac{c_1^2(1-\beta)}{120L\beta c_3^{1-\beta}}}^{\frac{1}{\beta}}, \frac{c_3c_1^2}{120}, \frac{1-\beta}{8\beta} \cdot \frac{c_3c_1^2}{480L} \}$.
    \end{itemize}
\end{lemma}
\begin{proof}
    We choose $c_3 = \frac{1}{2}$. We have $c_1 = \frac{1}{2} - F_{1,R}(\frac{1}{20L})$ is a constant.

    We note that
    \begin{align*}
        E_0 := F_{1,R}^{-1}(\frac{1}{1+\delta}) = \frac{\beta}{1-\beta} \InParentheses{(1+\delta)^{1-\beta} - \InParentheses{\frac{1+\delta}{\delta}}^{ 1-\beta}}
    \end{align*}
    satisfies $F_{1,R}(E_0) = \frac{1}{1+\delta}$. Similarly, we calculate
    \begin{align*}
        E_1 & := F_{1,R}^{-1}(\frac{1+c_3}{1+c_3 + \delta})                                                                                                                                       \\
            & = \frac{\beta}{1-\beta} \InParentheses{\InParentheses{\frac{1+c_3+\delta}{1+c_3}}^{1-\beta} - \InParentheses{\frac{1+ c_3 +\delta}{\delta}}^{ 1-\beta}}                             \\
            & \ge \frac{\beta}{1-\beta} \InParentheses{\InParentheses{1+\delta}^{1-\beta} - 2 - \InParentheses{\frac{1+ c_3 +\delta}{\delta}}^{ 1-\beta}}                                         \\
            & \ge \frac{\beta}{1-\beta} \InParentheses{\InParentheses{1+\delta}^{1-\beta}- \InParentheses{\frac{1+\delta}{\delta}}^{ 1-\beta} - \InParentheses{\frac{c_3}{\delta}}^{1-\beta} - 2} \\
            & = E_0 - \frac{\beta}{1-\beta}\InParentheses{\InParentheses{\frac{c_3}{\delta}}^{1-\beta} + 2}
    \end{align*}
    where in the first inequality we use the fact that $(1+\delta)^{1-\beta} \le 2$ since $\delta \le 1$; the second inequality we use the inequality $(x+y)^{1-\beta} \le x^{1-\beta} + y^{1-\beta}$. We note that
    \begin{align}\label{eq:delta conditon}
        \delta \le \delta_1:= \InParentheses{\frac{c_1^2(1-\beta)}{120L\beta c_3^{1-\beta}}}^{\frac{1}{\beta}} \Rightarrow - \frac{c_1^2}{30L\delta} \le - \frac{\beta}{1-\beta}\InParentheses{\InParentheses{\frac{c_3}{\delta}}^{1-\beta} + 2}.
    \end{align}
    Thus for any $\delta \le \delta_1$,  we have for any $E$ such that $F_{1,R}(E) \ge \frac{1}{1+\delta}$,
    \begin{align*}
        - \frac{c_1^2}{30L\delta} + E \le - \frac{c_1^2}{30L\delta} + E_0 \le E_0 - \frac{\beta}{1-\beta}\InParentheses{\InParentheses{\frac{c_3}{\delta}}^{1-\beta} + 2} \le E_1.
    \end{align*}
    The above implies $F_{1,R}(- \frac{c_1^2}{30L\delta} + E ) \ge \frac{1+c_3}{1+c_3 +\delta}$ and the first item in \Cref{assumption:main} is satisfied.

    We define $E_2$
    \begin{align*}
        E_2 := F_{1,R}^{-1}\InParentheses{\frac{1}{2(1+\delta)}} & = \frac{\beta}{1-\beta} \InParentheses{(2+2\delta)^{1-\beta} - \InParentheses{\frac{2+2\delta}{1+2\delta}}^{ 1-\beta}} \\
                                                                 & =  \frac{\beta}{1-\beta}(2+2\delta)^{1-\beta} \cdot \InParentheses{1 - \InParentheses{\frac{1}{1+2\delta}}^{ 1-\beta}} \\
                                                                 & \le \frac{4\beta}{1-\beta} \cdot \InParentheses{1 - \InParentheses{1 - \frac{2\delta}{1+2\delta}}^{1-\beta}}           \\
                                                                 & \le \frac{4\beta}{1-\beta} \cdot \frac{2\delta }{1+\delta}                                                             \\
                                                                 & = \frac{8\beta\delta}{(1-\beta)(1+\delta)}
    \end{align*}
    where in the first inequality we use $(2+2\delta)^{1-\beta} \le 4$ since $0\le\delta\le 1$ and $\beta \in (0,1)$; in the second inequality we use the basic inequality $(1-x)^r \le 1-x$ for $r, x\in(0,1)$. We define
    \begin{align*}
        \delta_2:= \min\{ \frac{c_3c_1^2}{120}, \frac{1-\beta}{8\beta} \cdot \frac{c_3c_1^2}{480L} \}
    \end{align*}
    Then for any $\delta \le \delta_2$ and $E$ such that $F_{1,R}[E] \ge \frac{1}{2(1+\delta)}$, we have
    \begin{align*}
        -\frac{c_3c_1^2}{120L} + \frac{\delta}{4L} + E & \le -\frac{c_3c_1^2}{120L} + \frac{\delta}{4L} + E_2                                          \\
                                                       & \le -\frac{c_3c_1^2}{120L} + \frac{c_3c_1^2}{480L} + \frac{8\beta\delta}{(1-\beta)(1+\delta)} \\
                                                       & \le -\frac{c_3c_1^2}{120L} + \frac{c_3c_1^2}{480L} + + \frac{c_3c_1^2}{480L}                  \\
                                                       & = -\frac{c_3c_1^2}{240L}.
    \end{align*}
    Thus we know $F_{1,R}(-\frac{c_3c_1^2}{120L} + \frac{\delta}{4L} + E) \ge F_{1,R}(-\frac{c_3c_1^2}{240L})$ and item 2 in \Cref{assumption:main} is satisfied by $c_2 =F_{1,R}(-\frac{c_3c_1^2}{240L}) - \frac{1}{2} > 0$.

    Combining the above, we can choose $\hat{\delta} = \min\{\delta_1, \delta_2\}$ so that both items in \Cref{assumption:main} hold for $\delta \le \hat{\delta}$.
\end{proof}

\section{Proof of \Cref{theorem:reduction from n to 2}}\label{app:high-dimension}
Recall the equivalent $2$-dimensional problem:
\begin{align*}
    \hat x^t & = \mathrm{Dupl}\left[\argmin_{x\in \frac{1}{n} \cdot \Delta^{2}} \left\{ \frac{n}{n^\alpha}\left\langle x, \sum_{\tau = 1}^{t-1} \ell^\tau + \ell^{t-1} \right\rangle + \frac{n}{\eta} R_1(x[1]) + \frac{n}{\eta}R_2(x[2]) \right\} \right] \\
             & = \mathrm{Dupl}\left[\frac{1}{n}\cdot \argmin_{x\in \Delta^{2}} \left\{ \frac{n}{n^\alpha}\left\langle \frac{1}{n} x, \sum_{\tau = 1}^{t-1} \ell^\tau + \ell^{t-1} \right\rangle + \frac{n}{\eta} R(x/n) \right\} \right]                   \\
             & = \mathrm{Dupl}\left[\frac{1}{n}\cdot \argmin_{x\in \Delta^{2}} \left\{ \left\langle x, \sum_{\tau = 1}^{t-1} \ell^\tau + \ell^{t-1} \right\rangle + \frac{n^{\alpha+1}}{\eta} R(x/n)  \right\} \right].
\end{align*}

\emph{Euclidean regularizer}: this regularizer is homogeneous of degree two. Choosing $\alpha = 1$, the inner minimization problem is exactly the same as the one solved by OFTRL in two dimensions.

\emph{Entropy regularizer}: we set $\alpha = 0$ to get equivalence:
\[
    n R(x/n) =  \sum_{i=1}^2 x[i] \log(x[i]/n) = \sum_{i=1}^2 x[i] \log x[i] - \sum_{i=1}^2 x[i] \log n = \sum_{i=1}^2 x[i] \log x[i] - \log n.
\]
Now we have equivalence because the last term is a constant that does not affect the $\argmin$.


\emph{Log regularizer:} we set $\alpha = -1$ to get equivalence, using similar logic as for entropy:
\[
    R(x/n) = \sum_{i=1}^2 -\log(x[i]/n) = 2\log n + \sum_{i=1}^2-\log x[i].
\]

\emph{Tsallis entropy regularizer:} we set $\alpha = -1+\beta$ to get equivalence, using similar logic as for entropy:
\[
    n^\beta R(x/n) = n^\beta \cdot \frac{1-\sum_{i=1}^2(\frac{x[i]}{n})^\beta}{1 -\beta} = \frac{n^\beta - 1}{1-\beta} + \frac{1-\sum_{i=1}^2 x[i]^\beta}{1 -\beta}.
\]

\notshow{

\section{Problem Constant-Independent Best-Iterate Convergence Rate for OMWU}
\label{app:OMWU best-iterate}
Our main results (\Cref{theorem: main} and \Cref{theorem: no last-iterate rate}) show that OMWU does not admit a last-iterate convergence rate that depends solely on $d_1$, $d_2$, $T$. The counterexample proving these negative results is $A_\delta \in [0,1]^{2 \times 2}$ parameterized by $\delta \in (0,\frac{1}{2})$ as defined in \eqref{eq:A delta}. Since $d_1 = d_2 = 2$ is constant, we show OMWU does not admit a last-iterate convergence rate that only depends on $T$, and a dependence on $1/\delta$ is inevitable.

In this section, we show that for the class of games $A_\delta$, OMWU enjoys $O(\frac{1}{\ln T})$ best-iterate convergence rate. We remark that although $O(\frac{1}{\ln T})$ is a slow convergence rate, it does not depend on $1/\delta$ and is much faster than the last-iterate rate, especially when $\delta \rightarrow 0$. The distinction between the best-iterate convergence rate and the last-iterate convergence rate of OMWU is interesting, as these two rates are comparable for OGDA. It remains an open question whether OMWU has a fast best-iterate convergence rate in general.

\begin{theorem}[Best-Iterate Convergence Rate of OMWU on $A_\delta$]
    For any $\delta \in (0,\frac{1}{32})$, OMWU with step size $\eta \le \frac{1}{8}$ on $A_\delta$ (defined in \eqref{eq:A delta}) satisfies for all $T \ge 2$,
    \begin{align*}
        \min_{t \in [T]} \DualGap(x^t, y^t) \le O\InParentheses{\frac{1}{\eta \ln T}}.
    \end{align*}
\end{theorem}
\begin{proof}
    We first present a sketch of the proof.
    \begin{itemize}
        \item[1.] We denote $T_1$ the first iteration when $x^t[1] \ge \frac{1}{1+\delta}$. We first show that OMWU has a linear convergence rate between $[1, T_1]$ and finally $\DualGap(x^{T_1}, y^{T_1}) \le \delta$. 
        \item[2.] Since at time $T_1$ the duality gap is only $O(\delta)$, the $O(\frac{1}{\ln T})$ best-iterate convergence rate holds until $T \ge \Omega(\exp(\frac{1}{\delta}))$. For all iterates after $T \ge \Omega(\exp(\frac{1}{\delta}))$, we will use the $\frac{\exp(O(1/\delta))}{\sqrt{T}}$ last-iterate convergence rate of OMWU \citep{wei2021linear}.

        \item[3.] Finally, we combine the convergence guarantees in the two phases to show a $\delta$-independent $O(\frac{1}{\ln T})$ best-iterate convergence rate for all $T \ge 2$.
    \end{itemize}

    We remark that OMWU has a closed-form update rule as follows:
    \begin{align*}
        x^t[1] & \propto x^1[1] \cdot \exp\InParentheses{-\eta L^{t-1}_x - \eta \ell^{t-1}_x}, \\
        y^t[1] & \propto y^1[1] \cdot \exp\InParentheses{-\eta L^{t-1}_y - \eta \ell^{t-1}_y}.
    \end{align*}
    \paragraph{Phase I: Linear Convergence}
    The OMWU dynamics starts with $(x^1, y^1)$ such that $x^1[1], y^1[1] = \frac{1}{2}$. Denote $T_0 = \lfloor \frac{1}{\eta} \rfloor + 1 < \frac{2}{\eta}$ and $T_1 \ge T_s$ the smallest iteration where $x^t[1] \ge \frac{1}{1+\delta}$.
    Using the update rule, we have
    \begin{align*}
        \frac{x^{T_0}[1]}{x^{T_0}[2]} & = \frac{x^1[1]}{x^1[2]} \cdot \exp\InParentheses{-\eta E^{T_0-1}_x - \eta e^{T_0-1}_x} \le e^{\frac{1}{2}\eta T_0} \le e,
    \end{align*}
    where we use $x^1[1] = x^1[2] = \frac{1}{2}$ and $-e^{t-1}_x \le \frac{1}{2}$ by \Cref{proposition:properties of A}. 

    For all $ t \in [1, T_0]$, since $x^t[1] \le \frac{e}{e+1}$ and $\delta \le \frac{1}{2e}$, we have $-e^t_y = -(\ell_y^{t}[1]- \ell_y^{t}[2])= (\frac{1}{2} + \delta)x^t[1] - 1 + \frac{1}{2}x^t[1] \le \frac{\delta e - 1}{e+1} \le - \frac{1}{2(e+1)}$. For all $t \in [1, T_1-1]$, since $x^t[1] \le \frac{1}{1+\delta}$, we have $ -e^t_y = -(\ell_y^{t}[1]- \ell_y^{t}[2]) = (\frac{1}{2} + \delta)x^t[1] - 1 + \frac{1}{2}x^t[1] \le 0$.

    Let us consider an auxiliary sequence $\{\tx^t, \ty^t\}$ defined as the vanilla MWU algorithm with loss vectors $L_x^{t-1}$ and $L_y^{t-1}$ as follows: for $i \in \{1,2\}$ and $t \ge 1$
    \begin{align*}
        \tx^t[i] & \propto x^1[i] \cdot \exp\{-\eta L^{t-1}_x[i]\} \\
        \ty^t[i] & \propto y^1[i] \cdot \exp\{-\eta L^{t-1}_y[i]\}
    \end{align*}
    It is clear that
    \begin{align*}
        \frac{x^t[1]}{x^t[2]} = \frac{\tx^t[1]}{\tx^t[2]} \cdot \exp\{ -\eta e^{t-1}_x\}, \quad \frac{y^t[1]}{y^t[2]} = \frac{\ty^t[1]}{\ty^t[2]} \cdot \exp\{ -\eta e^{t-1}_y\}.
    \end{align*}

    Now for any $t \in [T_0, T_1]$, we have
    \begin{align*}
        \frac{\ty^t[1]}{\ty^t[2]} & =  \frac{\ty^1[1]}{\ty^1[2]} \cdot \exp\InParentheses{-\eta E^{T_0-1}_y}                                               \\
                                  & =   \exp\InParentheses{-\eta E^{T_0-1}_y}                                                                              \\
                                  & \le \exp\InParentheses{-\frac{\eta (T_0 - 1)}{2(e+1)}} \tag{$-(e_y^{t} \le \frac{1}{2(e+1)}$ for all $t \in [1, T_0]$} \\
                                  & \le \exp\InParentheses{-\frac{1}{2(e+1)}} < \frac{8}{9}. \tag{$T_0-1 \ge \frac{1}{\eta}$}
    \end{align*}
    Then we have for any $t \in [T_0, T_1]$,
    \begin{align*}
        \frac{y^{t}[1]}{y^{t}[2]} & = \frac{\ty^1[1]}{\ty^1[2]} \cdot \exp\InParentheses{- \eta e^{t-1}_y}                         \\
                                  & \le \frac{\ty^1[1]}{\ty^1[2]} < \frac{8}{9}. \tag{$-e^{t-1}_y \le 0$ for all $t \in [1, T_1]$}
    \end{align*}
    This implies $y^{t}[1] < \frac{8}{17}$ for all $t \in [T_0, T_1]$. Moreover, for all $t \in [T_0, T_1]$, we have
    \begin{align*}
        e^t_x = \ell^t_x[1] - \ell^t_x[2] & = \frac{1}{2} + \delta y^t[1] - 1 + y^t [1]                               \\
                                          & \le -\frac{1}{2} + \frac{8}{17}(1+\delta) \tag{$y^t[1] \le \frac{8}{17}$} \\
                                          & \le -\frac{1}{68} \tag{$\delta \le \frac{1}{32}$}.
    \end{align*}

    Moreover, since $e^t_x = \ell^t_x[1] - \ell^t_x[2] \le \frac{1}{2} + \delta \le 1$ always holds, we have for all $t \in [1, T_0+1]$,
    \begin{align*}
        \frac{x^{t}[1]}{x^{t}[2]} & \ge \frac{x^1[1]}{x^1[2]}\cdot \exp\InParentheses{- \eta E^{t-1}_x -\eta e^{t-1}_x}         \\
                                  & \ge  e^{-\eta t}                                                                            \\
                                  & \ge  e^{-2} \ge \frac{1}{9} \tag{$t \le T_0+1 < \frac{2}{\eta}$ as $\eta \le \frac{1}{8}$}.
    \end{align*}

    Combing the above, we get for all $t \in [T_0+1, T_1]$,
    \begin{align*}
        \frac{x^t[1]}{x^t[2]} & \ge \frac{x^{T_0}[1]}{x^{T_0}[2]}\cdot \exp\InParentheses{- 2\eta e^{t-1}_x - \sum_{k=T_0}^{t-2} e^k_x + \eta e^{T_0-1}_x}                                      \\
                              & \ge \frac{1}{9} \cdot \exp\InParentheses{ \frac{\eta(t- T_0)}{68} - \frac{\eta}{2}} \tag{$-\frac{1}{2} \le e^{t-1}_x \le -\frac{1}{68}$ for $t \in [T_0, T_1]$} \\
                              & \ge \frac{1}{9} \cdot \exp\InParentheses{ \frac{\eta(t- T_0)}{68}-\frac{1}{16}}. \tag{$\eta \ge \frac{1}{8}$}                                                   \\
                              & \ge \frac{1}{10} \cdot \exp\InParentheses{ \frac{\eta(t- T_0)}{68}}.
    \end{align*}

    Now we track the duality gap. Note that for $t \in [T_0, T_1-1]$, we have $x^t[1] \le \frac{1}{1+\delta}$ and $y^1[1] \le \frac{8}{17} \le \frac{1}{2(1+\delta)}$. Therefore,
    \begin{align*}
        \DualGap(x^t, y^t) & = \max_{i\in \{1,2\}} (A_{\delta}^\top x^t)[i] - \min_{i\in \{1,2\}} (A_\delta y^t)[i] \\
                           & = 1 -\frac{1}{2}x^t[1] - \frac{1}{2} - \delta y^t[1]                                   \\
                           & \le \frac{1}{2}(1-x^t[1])                                                              \\
                           & = \frac{1}{2(\frac{x^t[1]}{x^t[2]}+1)}                                                 \\
                           & \le  \frac{1}{2} \cdot \frac{x^t[2]}{x^t[1]}
    \end{align*}
    Then we get for all $t \in [T_0, T_1-1]$,
    \begin{align*}
        \DualGap(x^t, y^t) & \le 5\cdot \exp\InParentheses{ -\frac{\eta(t- T_0)}{68}}                                                                    \\
                           & \le 6\cdot \exp\InParentheses{-\frac{\eta t}{68}} \tag{$T_0 < \frac{2}{\eta}$ and $\exp(\frac{1}{34}) \le 1 + \frac{1}{5}$}
    \end{align*}
    Since $T_0 < \frac{2}{\eta}$, we can conclude that for all $t \in [1, T_1]$, $\DualGap(x^t, y^t)\le 6 \cdot \exp\InParentheses{-\frac{\eta t}{68}}$. Moreover, since $x^{T_1}[1] \ge \frac{1}{1+\delta}$ and $y^{T_1}[1] < \frac{8}{17} \le \frac{1}{2(1+\delta)}$, we have
    \begin{align*}
        \DualGap(x^{T_1}, y^{T_1}) & = \max_{i\in \{1,2\}} (A_{\delta}^\top x^{T_1})[i] - \min_{i\in \{1,2\}} (A_\delta y^{T_1})[i] \\
                                   & = (\frac{1}{2} + \delta)x^{T_1}[1] - \frac{1}{2} - \delta y^{T_1}[1]                           \\
                                   & \le \delta.
    \end{align*}
    To conclude, for each step $t \in [1,T_1]$, we have (1) a linear convergence rate $\DualGap(x^t, y^t) \le 6\exp(-\frac{\eta t}{68})$; (2)  $\DualGap(x^{T_1}, y^{T_1}) \le \delta$.

    \paragraph{Phase II: Sublinear Convergence with  Dependence on $1/\delta$} We defer the analysis to \Cref{app:sublinear OMWU}. By \Cref{corollary:OMWU sublinear}, we know for all $t \ge 1$,
    \begin{align*}
        \DualGap(x^t, y^t) \le \frac{1200e^{\frac{10}{\delta}}}{\eta} \cdot \frac{1}{\sqrt{t}}.
    \end{align*}
    Note that although this rate is universal for the last iterate, it has an exponential dependence on $1/\delta$. Thus the last-iterate convergence rate is meaningful only after $e^{\Omega(1/\delta)}$ iterations.

    \paragraph{Combining Both Phases for $\delta$-independent Best-Iterate Convergence Rate}
    Now we show how to combine the two convergence guarantees in different phases to get a $\delta$-independent $O(\frac{1}{\ln T})$ best-iterate convergence rate.
    For all $T \in [2, T_1]$, we can use the linear convergence rate $\DualGap(x^t, y^t) \le 6\exp(-\frac{\eta t}{68})$. We choose a constant $C_\eta \ge 1$ such that
    \begin{align*}
        f(T):= \frac{C_\eta}{\eta \ln T} \le 6 \exp\InParentheses{-\frac{\eta T}{68}}, \forall T \ge 1.
    \end{align*}
    Thus, the $O(\frac{1}{\eta \ln T})$ best-iterate convergence rate holds for all $T \le T_1$.

    We also note that for all $T \ge e^{\frac{36}{ \delta}}$, we have
    \begin{align*}
        \frac{1200e^{\frac{10}{\delta}}}{\eta} \cdot \frac{1}{\sqrt{T}} & = \frac{1}{\eta \ln T} \cdot \frac{1200 e^{\frac{10}{\delta}} \ln T}{\sqrt{T}}                                                         \\
        & \le \frac{1}{\eta \ln T} \cdot \frac{1200 e^{\frac{10}{\delta}}}{T^{1/3}} \tag{$\frac{\ln T}{\sqrt{T}} \le \frac{1}{T^{\frac{1}{3}}}$} \\
        & \le \frac{1}{\eta \ln T} \cdot \frac{1200 e^{\frac{10}{\delta}}}{e^{\frac{12}{\delta}}}                                                \\
        \tag{$\delta \le \frac{1}{32}$}
        & \le  \frac{1}{\eta \ln T}.
    \end{align*}
    Thus $O(\frac{1}{\eta \ln T})$ best-iterate convergence rate holds for $T \ge e^{\frac{36}{\delta}}$.

    For other iterate $T \ge T_1$ but less than $e^{\frac{36}{\delta}}$, we know
    \begin{align*}
        \min_{t \in [T]}\DualGap(x^T, y^T) & \le \DualGap(x^{T_1}, y^{T_1})                           \\
                                           & \le \delta                                               \\
                                           & \le \frac{36}{\ln T} \tag{$T \le e^{\frac{36}{\delta}}$}.
    \end{align*}
    Thus we know the $O(\frac{1}{\eta \ln  T})$ best-iterate convergence rate holds for all $T \ge T_1$.

    In conclusion, OMWU enjoys the following best-iterate convergence rate for all $T \ge 2$
    \begin{align*}
        \min_{t \in [T]} \DualGap(x^t, y^t) \le O\InParentheses{\frac{1}{\eta \ln  T}}.
    \end{align*}
\end{proof}

\subsection{Existing Results from \citep{wei2021linear}}
\label{app:sublinear OMWU}
We consider a matrix game $A \in [0,1]^{d_1 \times d_2}$ with a unique Nash equilibrium $z^* = (x^*, y^*)$. Below we define some problem-dependent constants. We recall that  $\+X = \Delta^{d_1}$ and $\+Y = \Delta^{d_2}$ and define
\begin{align*}
     & \+V^*(\+X) := \left\{x : x \in \+X, \supp(x) \subseteq \supp(x^*) \right\},  \\
     & \+V^*(\+Y) := \left\{y : y \in \+Y, \supp(y) \subseteq \supp(y^*) \right\} .
\end{align*}
\begin{definition}
    \label{dfn:c}
    Define $c = \min\{c_x, c_y\}$ where
    \begin{align*}
        c_x:= \min_{x \in \+X \setminus \{x^*\}} \max_{y \in \+V^*(\+Y)} \frac{(x-x^*)^\top A y}{\InNorms{x - x^*}_1}, \quad c_y:= \min_{y \in \+Y \setminus \{y^*\}} \max_{x \in \+V^*(\+X)} \frac{x^\top A (y^* - y)}{\InNorms{y^* - y}_1}.
    \end{align*}
\end{definition}
\begin{definition}
    \label{dfn:eps}
    Define $\epsilon$ as
    \begin{align*}
        \epsilon := \min_{j \in \supp(z^*)}  \exp\InParentheses{- \frac{\ln(d_1 d_2)}{z_{j^*}}}.
    \end{align*}
\end{definition}
In the analysis, we sometimes use $z = (x, y) \in \+X \times \+Y$ to simplify the notation. We denote $\KL(x,x')$ the Kullback–Leibler (KL) divergence. We slightly abuse the notation and denote $\KL(z,z') = \KL(x,x') + \KL(y, y')$.
\begin{theorem}[Adapted from Lemma 2 and Theorem 3 in \citep{wei2021linear}]
    \label{thm:sublinear-OMWU-old}
    Assume the game $A \in [0,1]^{d_1 \times d_2}$ has a unique Nash equilibrium. Then OMWU with step size $\eta \le \frac{1}{8}$ on $A$ satisfies for all $T \ge 1$,
    \begin{align*}
        \DualGap(x^T, y^T) \le \sqrt{d_1d_2} \cdot \KL(z^*,z^T) \le \frac{C \sqrt{d_1d_2}}{\eta \sqrt{T}},
    \end{align*}
    where $C := \frac{8\sqrt{128(128+2\ln(d_1d_2))}}{\sqrt{15}} \cdot \frac{1}{\epsilon^3\cdot c}$.
\end{theorem}
\begin{proof}
    Define $F(z) = (Ay, -A^\top x)$. We note that since $A \in [0,1]^{d_1 \times d_2}$, its gradient norms $\InNorms{Fz}_2$ are bounded by $\sqrt{d_1d_2}$. Then we have
    \begin{align*}
        \DualGap(z^t) & \le \max_{z \in \+X \times \Y} \InAngles{F(z), z^t - z}                            \\
                      & = \max_{z \in \+X \times \Y} \InAngles{F(z), z^* - z} + \InAngles{F(z), z^t - z^*} \\
                      & \le \max_{z \in \+X \times \Y}\InAngles{F(z) , z^t - z^*}                          \\
                      & \le\max_{z \in \+X \times \Y} \InNorms{F(z)}_2 \InNorms{z^t - z^*}_2               \\
                      & \le \sqrt{d_1d_2} \KL(z^*,z^t),
    \end{align*}
    where in the second inequality, we use the fact that $z^*$ is a Nash equilibrium; in the third inequality, we use the triangle inequality; and in the last inequality, we use $\InNorms{Fz}_2 \le \sqrt{d_1d_2}$. The rest of the proof follows from the proof of Lemma 2 and Theorem 3 in \citep{wei2021linear}, where they give a bound on $\KL( z^*,z^t)$.
\end{proof}

\paragraph{Constants for $A_\delta$} By \Cref{proposition:properties of A}, we know $A_\delta$ has a unique Nash equilibrium $z^* = (x^* = (\frac{1}{1+\delta}, \frac{\delta}{1+\delta}), y^* = (\frac{1}{2(1+\delta)}, \frac{1+2\delta}{2(1+\delta)}))$. Now we calculate the parameter $C$ for $A_\delta$. We first calculate $c_x$ and $c_y$.
\begin{proposition}
    For $\delta \in (0,\frac{1}{2})$ and $A_\delta$ defined in \eqref{eq:A delta}, $c_x$ and $c_y$ defined in \Cref{dfn:c} satisfies
    \begin{align*}
        c_x = \frac{1}{4}, \quad c_y = \frac{\delta}{2}.
    \end{align*}
    Hence, $c:= \min\{c_x, c_y\} = \frac{\delta}{2}$.
\end{proposition}
\begin{proof}
    By \Cref{dfn:c}, we have
    \begin{align*}
        c_x & = \min_{x\in \+X \setminus \{x^*\}} \max_{y \in \+V^*(\+Y)} \frac{(x-x^*)^\top A y}{\InNorms{x - x^*}}                                                                                                            \\
            & = \min_{x[1] \in [0,1], x[1] \ne \frac{1}{1+\delta}} \max_{y[1] \in [0,1]} \frac{(1+\delta)y[1]-\frac{1}{2}}{2} \cdot \frac{x[1] - \frac{1}{1+\delta}}{\left | x[1] - \frac{1}{1+\delta} \right|}  = \frac{1}{4}.
    \end{align*}
    Similarly, for $c_y$ we have
    \begin{align*}
        c_y & = \min_{y \in \+Y \setminus \{y^*\}} \max_{x \in \+V^*(\+X)} \frac{x^\top A (y^* - y)}{\InNorms{y^* - y}_1}                                                                                                     \\
            & = \min_{y[1] \in [0,1],y[1] \ne \frac{1}{2(1+\delta)}} \max_{x[1] \in [0,1]} \frac{(1+\delta)x[1] - 1}{2} \cdot \frac{\frac{1}{2(1+\delta)}-y[1]}{\left |\frac{1}{2(1+\delta)}-y[1] \right|} =\frac{\delta}{2}.
    \end{align*}
    The completes the proof.
\end{proof}
For $\epsilon$ defined in \Cref{dfn:eps}, we can easily get $\epsilon = \exp(-\frac{\ln4}{\frac{\delta}{1+\delta}}) \ge e^{-\frac{3}{\delta}}$ since $(1+\delta)\ln 4 \le 3$.

Plugging $c = \frac{\delta}{2}$ and $\epsilon \ge e^{-\frac{3}{\delta}}$ into \Cref{thm:sublinear-OMWU-old}, we get the following last-iterate convergence rate of OMWU on $A_\delta$.
\begin{corollary}
    \label{corollary:OMWU sublinear}
    For any $\delta \in (0,\frac{1}{2})$, OMWU with step size $\eta \le \frac{1}{8}$ on $A_\delta$ satisfies for all $T \ge 1$,
    \begin{align*}
        \DualGap(x^T, y^T) \le \frac{1200e^{\frac{10}{\delta}}}{\eta} \cdot \frac{1}{\sqrt{T}},
    \end{align*}
\end{corollary}
\begin{proof}
    By \Cref{thm:sublinear-OMWU-old}, we have
    \begin{align*}
        \DualGap(z^t) \le \frac{C\sqrt{d_1d_2}}{\eta \sqrt{T}} = \frac{2C}{\eta \sqrt{T}},
    \end{align*}
    where
    \begin{align*}
        C & = \frac{8\sqrt{128(128+2\ln(d_1d_2))}}{\sqrt{15}} \cdot \frac{1}{\epsilon^3\cdot c} \\
          & \le 300 \cdot \frac{1}{e^{-\frac{9}{\delta}} \cdot \frac{\delta}{2} }                 \\
          & = \frac{600 e^{\frac{9}{\delta}}}{\delta}                                                 \\
          & \le 600 e^{\frac{10}{\delta}}\tag{$x \le e^x$}
    \end{align*}
    This completes the proof.
\end{proof}}

\notshow{
\subsection{Improved Best-Iterate Convergence of OMWU}
In this section, we use the OMD-type update rule of OMWU, which is equivalent to the OFTRL-style update rule. Let $z = (x,y)$ and initial points $z^1 = \hz^1$ be the uniform distribution. Let $F(z^t) = (\ell^t_x, \ell^t_y) = (Ay^t, -A^\top x^t)$. Then for $t \ge 2$,
\begin{align*}
    \hz^t &= \argmin_{z \in \+Z} \{ \eta \InAngles{z, F(z^{t-1})} + D(z, \hz^{t-1}) \} \\ 
    z^t &= \argmin_{z \in \+Z} \{ \eta \InAngles{z, F(z^{t-1})} + D(z, \hz^t) \} \\ 
    \hz^{t+1} &= \argmin_{z \in \+Z} \{ \eta \InAngles{z, F(z^t)} + D(z, \hz^t) \} 
\end{align*}
Let $z^*$ be the Nash equilibrium. Define $\Theta^t = \KL(z^*, \hz^t) + \frac{1}{16} \KL(\hz^t, z^{t-1})$ and $\zeta^t = \KL(\hz^{t+1}, z^t) + \KL(z^t, \hz^t)$.
\begin{lemma}[Adapted from \citep{wei2021linear}] Consider OWMU for a zero-sum game $A \in [0,1]^{d\times d}$ with $\eta \le \frac{1}{8}$. Then, for any $t \ge 2$, it holds that.
    \begin{align*}
        \Theta^{t+1} \le \Theta^t - \frac{15}{16}\zeta^t.
    \end{align*}
\end{lemma}
Define $C_1 = \frac{16}{15}\Theta^2$. Then we have $\Theta^t \le C_1$ for all $t \ge 2$ and $\sum_{t=2}^\infty \zeta^t \le C_1$.

\begin{claim}
    For all $t \ge 2$, it holds that $y^t[1], \hy^t[1] \in [c, 1-c]$ for some constant $c$ that dose not depend on $\delta$.
\end{claim}
\begin{proof}
    Since $\KL(y^*, \hy^t) \le C_1$ and $y^* \approx (\frac{1}{2}, \frac{1}{2})$. Plus stability between $y^t$ and $\hy^t$.
\end{proof}
\begin{claim}
    For all $t \ge 2$, it holds that $\hx^t[1] \ge c > 0$ for some constant $c$ that dose not depend on $\delta$.
\end{claim}
\begin{proof}
    Since $\KL(x^*, \hx^t) \le C_1$ and $x^*[1] \ge \frac{2}{3}$.
\end{proof}
The only bad case is that $\hx^t[1]$ could be very close to $1$ (equivalently, $\hx^t[2] = 1 - \hx^t[1]$ could be very small). We will show that $\hx^t[1] \ge 1 -\frac{\delta}{3}$ happens for at most $O(\frac{1}{\delta})$ times.

\begin{proposition}
    Let $p \in (0,1)$ and $x \in [0,1]$, then it holds that 
    \begin{align*}
        \frac{p \cdot e^x}{p \cdot e^x + 1 - p} - p \ge \frac{x p}{8}.
    \end{align*}
\end{proposition}
\begin{proof}
    Using inequality $e^x -1 \ge x$ and the fact $e^x \le 3$, we have 
    \begin{align*}
        \frac{p \cdot e^x}{p \cdot e^x + 1 - p} - p = \frac{ p(1-p)(e^x - 1)}{p \cdot e^x + 1 - p} \ge x \cdot \frac{ p(1-p)}{p \cdot e^x + 1 - p} \ge x \cdot \frac{ p(1-p)}{4} \ge \frac{xp}{8}.
    \end{align*}
\end{proof}

\begin{proposition}
    Let $\ell \in \-R^2$ such that $q:= |\ell[1] - \ell[2]| \le 1$, and $v \in \Delta^2$ with $p := \min_i v[i]$. Let $u = \argmin_{x \in \Delta^2}\{\ell^\top x + D(x, v)\}$. Then $\KL(u,v), \KL(v,u) \ge \frac{q p}{4}$.
\end{proposition}
\begin{proof}
    Without loss of generality, we assume $\ell[1] \le \ell[2]$. Then
    \begin{align*}
        u[1] - v[1] = \frac{v[1] \cdot e^{q}}{v[1]\cdot e^{q} + 1 - v[1]} - v[1] \ge \frac{q \cdot v[1] \cdot (1- v[1])}{4} \ge \frac{qp(1-p)}{4} \ge \frac{qp}{8}.
    \end{align*}
    Then $\KL(u,v), \KL(v,u) \ge \InNorms{u-v}_1 = 2(u[1] - v[1]) \ge \frac{qp}{4}$.
\end{proof}

\begin{lemma}
    If $x^t[1] \ge 1- \frac{\delta}{3}$, then $\KL(y^{t+1}, \hy^{t+1}) \ge \frac{c\eta \delta}{12}$.
\end{lemma}
\begin{proof}
    Since $x^t[1] \ge 1- \frac{\delta}{3}$, we can use \Cref{proposition:properties of A} to verify that $e^t_y = \ell^t_y[1] - \ell^t_y[2] \le -\frac{\delta}{3}$. Since $z^{t+1}$ is updated from $\hz^{t+1}$ using $\ell^t_y$, we have $\KL(y^{t+1}, \hy^{t+1}) \ge \frac{\eta\delta}{3} \cdot \frac{c}{4}  = \frac{c\eta \delta}{12}$.
\end{proof}

\begin{theorem}
    There are at most $\frac{12C_1}{c \eta \delta}$ iterations where $x^t[1] \ge 1 - \frac{\delta}{6}$ or $\hx^t[1]\ge 1 - \frac{\delta}{6}$. 
\end{theorem}
\begin{proof}
    Using the above lemma and the stability between $x^t$ and $\hx^t$.
\end{proof}

Then we can conclude a $O(\frac{1}{\delta\sqrt{T}})$ best-iterate convergence rate for $T \ge \Omega(\frac{1}{\delta})$. Combining this rate with the initial phase of linear convergence to an $O(\delta)$-approximate Nash equilibrium, we get a $O(T^{-\frac{1}{4}})$ global best-iterate convergence rate of OMWU.}

\section{Numerical Experiments with Adaptive Stepsizes}\label{app:simu adaptive}
In this section we present our numerical results when \ref{OFTRL} and \ref{OOMD} are instantiated with adaptive stepsize~\citep{duchi2011adaptive}: $\eta_{t} = 1 / \sqrt{\epsilon+\sum_{k=1}^{t-1} \| \bm{\ell}_k\|_{k}^{2}}$ with some constant $\epsilon>0$. We present our numerical experiments in Figure \ref{fig:plots with adaptive step size}, where we choose $\epsilon = 0.1$.

\begin{figure}[t]
    \includegraphics[scale=.72]{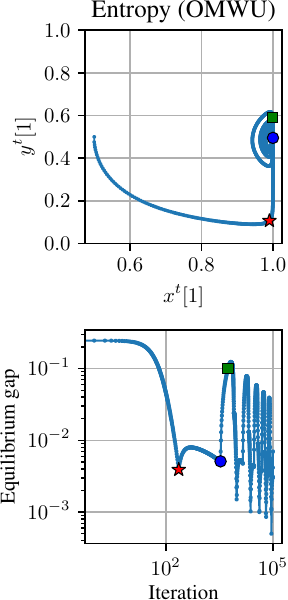}\hfill%
    \includegraphics[scale=.72]{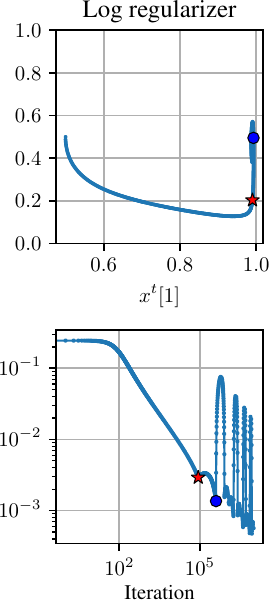}\hfill%
    \includegraphics[scale=.72]{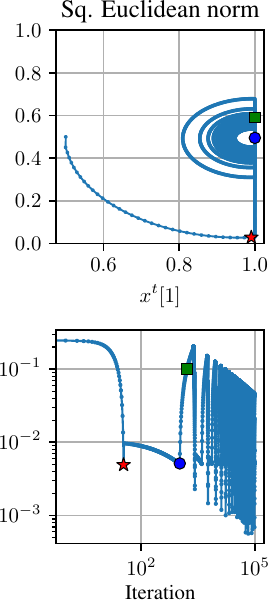}\hfill%
    \includegraphics[scale=.72]{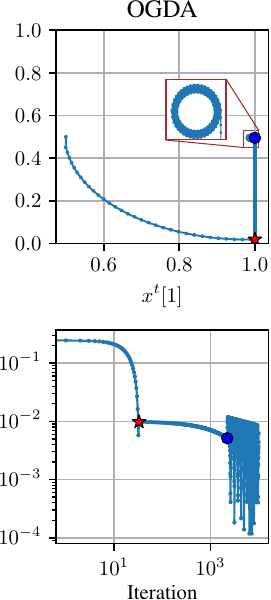}%
    \caption{
        Comparison of the dynamics produced by three variants of OFTRL with different regularizers (negative entropy, logarithmic regularizer, and squared Euclidean norm) and OGDA in the same game $A_\delta$ defined in \eqref{eq:A delta} for $\delta := 10^{-2}$ and adaptive step size with $\epsilon = 0.1$. The bottom row shows the duality gap achieved by the iterates.
    }
    \label{fig:plots with adaptive step size}
\end{figure}

\iftrue
    \newpage
    \section*{NeurIPS Paper Checklist}

    \begin{enumerate}

        \item {\bf Claims}
        \item[] Question: Do the main claims made in the abstract and introduction accurately reflect the paper's contributions and scope?
        \item[] Answer: \answerYes{} 
        \item[] Justification:  The claims made in the abstract and introduction match both the theoretical and experimental results of the paper.
        \item[] Guidelines:
            \begin{itemize}
                \item The answer NA means that the abstract and introduction do not include the claims made in the paper.
                \item The abstract and/or introduction should clearly state the claims made, including the contributions made in the paper and important assumptions and limitations. A No or NA answer to this question will not be perceived well by the reviewers.
                \item The claims made should match theoretical and experimental results, and reflect how much the results can be expected to generalize to other settings.
                \item It is fine to include aspirational goals as motivation as long as it is clear that these goals are not attained by the paper.
            \end{itemize}

        \item {\bf Limitations}
        \item[] Question: Does the paper discuss the limitations of the work performed by the authors?
        \item[] Answer: \answerYes{} 
        \item[] Justification:  The extension to the general $m \times n$ game does not work for the Euclidean and log regularizers because the rescaling factors would be different for the row and column players.
        \item[] Guidelines:
            \begin{itemize}
                \item The answer NA means that the paper has no limitation while the answer No means that the paper has limitations, but those are not discussed in the paper.
                \item The authors are encouraged to create a separate "Limitations" section in their paper.
                \item The paper should point out any strong assumptions and how robust the results are to violations of these assumptions (e.g., independence assumptions, noiseless settings, model well-specification, asymptotic approximations only holding locally). The authors should reflect on how these assumptions might be violated in practice and what the implications would be.
                \item The authors should reflect on the scope of the claims made, e.g., if the approach was only tested on a few datasets or with a few runs. In general, empirical results often depend on implicit assumptions, which should be articulated.
                \item The authors should reflect on the factors that influence the performance of the approach. For example, a facial recognition algorithm may perform poorly when image resolution is low or images are taken in low lighting. Or a speech-to-text system might not be used reliably to provide closed captions for online lectures because it fails to handle technical jargon.
                \item The authors should discuss the computational efficiency of the proposed algorithms and how they scale with dataset size.
                \item If applicable, the authors should discuss possible limitations of their approach to address problems of privacy and fairness.
                \item While the authors might fear that complete honesty about limitations might be used by reviewers as grounds for rejection, a worse outcome might be that reviewers discover limitations that aren't acknowledged in the paper. The authors should use their best judgment and recognize that individual actions in favor of transparency play an important role in developing norms that preserve the integrity of the community. Reviewers will be specifically instructed to not penalize honesty concerning limitations.
            \end{itemize}

        \item {\bf Theory Assumptions and Proofs}
        \item[] Question: For each theoretical result, does the paper provide the full set of assumptions and a complete (and correct) proof?
        \item[] Answer: \answerYes{} 
        \item[] Justification: See \Cref{sec: pre}, \Cref{sec:analysis} and the appendix.
        \item[] Guidelines:
            \begin{itemize}
                \item The answer NA means that the paper does not include theoretical results.
                \item All the theorems, formulas, and proofs in the paper should be numbered and cross-referenced.
                \item All assumptions should be clearly stated or referenced in the statement of any theorems.
                \item The proofs can either appear in the main paper or the supplemental material, but if they appear in the supplemental material, the authors are encouraged to provide a short proof sketch to provide intuition.
                \item Inversely, any informal proof provided in the core of the paper should be complemented by formal proofs provided in appendix or supplemental material.
                \item Theorems and Lemmas that the proof relies upon should be properly referenced.
            \end{itemize}

        \item {\bf Experimental Result Reproducibility}
        \item[] Question: Does the paper fully disclose all the information needed to reproduce the main experimental results of the paper to the extent that it affects the main claims and/or conclusions of the paper (regardless of whether the code and data are provided or not)?
        \item[] Answer: \answerYes{} 
        \item[] Justification: See \Cref{fig:intro plots} and Python codes in the supplementary materials.
        \item[] Guidelines:
            \begin{itemize}
                \item The answer NA means that the paper does not include experiments.
                \item If the paper includes experiments, a No answer to this question will not be perceived well by the reviewers: Making the paper reproducible is important, regardless of whether the code and data are provided or not.
                \item If the contribution is a dataset and/or model, the authors should describe the steps taken to make their results reproducible or verifiable.
                \item Depending on the contribution, reproducibility can be accomplished in various ways. For example, if the contribution is a novel architecture, describing the architecture fully might suffice, or if the contribution is a specific model and empirical evaluation, it may be necessary to either make it possible for others to replicate the model with the same dataset, or provide access to the model. In general. releasing code and data is often one good way to accomplish this, but reproducibility can also be provided via detailed instructions for how to replicate the results, access to a hosted model (e.g., in the case of a large language model), releasing of a model checkpoint, or other means that are appropriate to the research performed.
                \item While NeurIPS does not require releasing code, the conference does require all submissions to provide some reasonable avenue for reproducibility, which may depend on the nature of the contribution. For example
                      \begin{enumerate}
                          \item If the contribution is primarily a new algorithm, the paper should make it clear how to reproduce that algorithm.
                          \item If the contribution is primarily a new model architecture, the paper should describe the architecture clearly and fully.
                          \item If the contribution is a new model (e.g., a large language model), then there should either be a way to access this model for reproducing the results or a way to reproduce the model (e.g., with an open-source dataset or instructions for how to construct the dataset).
                          \item We recognize that reproducibility may be tricky in some cases, in which case authors are welcome to describe the particular way they provide for reproducibility. In the case of closed-source models, it may be that access to the model is limited in some way (e.g., to registered users), but it should be possible for other researchers to have some path to reproducing or verifying the results.
                      \end{enumerate}
            \end{itemize}

        \item {\bf Open access to data and code}
        \item[] Question: Does the paper provide open access to the data and code, with sufficient instructions to faithfully reproduce the main experimental results, as described in supplemental material?
        \item[] Answer: \answerYes{} 
        \item[] Justification: See Python code in the supplementary materials.
        \item[] Guidelines:
            \begin{itemize}
                \item The answer NA means that paper does not include experiments requiring code.
                \item Please see the NeurIPS code and data submission guidelines (\url{https://nips.cc/public/guides/CodeSubmissionPolicy}) for more details.
                \item While we encourage the release of code and data, we understand that this might not be possible, so “No” is an acceptable answer. Papers cannot be rejected simply for not including code, unless this is central to the contribution (e.g., for a new open-source benchmark).
                \item The instructions should contain the exact command and environment needed to run to reproduce the results. See the NeurIPS code and data submission guidelines (\url{https://nips.cc/public/guides/CodeSubmissionPolicy}) for more details.
                \item The authors should provide instructions on data access and preparation, including how to access the raw data, preprocessed data, intermediate data, and generated data, etc.
                \item The authors should provide scripts to reproduce all experimental results for the new proposed method and baselines. If only a subset of experiments are reproducible, they should state which ones are omitted from the script and why.
                \item At submission time, to preserve anonymity, the authors should release anonymized versions (if applicable).
                \item Providing as much information as possible in supplemental material (appended to the paper) is recommended, but including URLs to data and code is permitted.
            \end{itemize}

        \item {\bf Experimental Setting/Details}
        \item[] Question: Does the paper specify all the training and test details (e.g., data splits, hyperparameters, how they were chosen, type of optimizer, etc.) necessary to understand the results?
        \item[] Answer: \answerYes{} 
        \item[] Justification: See Python code in the supplementary materials.
        \item[] Guidelines:
            \begin{itemize}
                \item The answer NA means that the paper does not include experiments.
                \item The experimental setting should be presented in the core of the paper to a level of detail that is necessary to appreciate the results and make sense of them.
                \item The full details can be provided either with the code, in appendix, or as supplemental material.
            \end{itemize}

        \item {\bf Experiment Statistical Significance}
        \item[] Question: Does the paper report error bars suitably and correctly defined or other appropriate information about the statistical significance of the experiments?
        \item[] Answer: \answerNA{} 
        \item[] Justification: There is no randomness in our experiments.
        \item[] Guidelines:
            \begin{itemize}
                \item The answer NA means that the paper does not include experiments.
                \item The authors should answer "Yes" if the results are accompanied by error bars, confidence intervals, or statistical significance tests, at least for the experiments that support the main claims of the paper.
                \item The factors of variability that the error bars are capturing should be clearly stated (for example, train/test split, initialization, random drawing of some parameter, or overall run with given experimental conditions).
                \item The method for calculating the error bars should be explained (closed form formula, call to a library function, bootstrap, etc.)
                \item The assumptions made should be given (e.g., Normally distributed errors).
                \item It should be clear whether the error bar is the standard deviation or the standard error of the mean.
                \item It is OK to report 1-sigma error bars, but one should state it. The authors should preferably report a 2-sigma error bar than state that they have a 96\% CI, if the hypothesis of Normality of errors is not verified.
                \item For asymmetric distributions, the authors should be careful not to show in tables or figures symmetric error bars that would yield results that are out of range (e.g. negative error rates).
                \item If error bars are reported in tables or plots, The authors should explain in the text how they were calculated and reference the corresponding figures or tables in the text.
            \end{itemize}

        \item {\bf Experiments Compute Resources}
        \item[] Question: For each experiment, does the paper provide sufficient information on the computer resources (type of compute workers, memory, time of execution) needed to reproduce the experiments?
        \item[] Answer: \answerYes{} 
        \item[] Justification: Our simulations were done within a few hours on an average consumer laptop.
        \item[] Guidelines:
            \begin{itemize}
                \item The answer NA means that the paper does not include experiments.
                \item The paper should indicate the type of compute workers CPU or GPU, internal cluster, or cloud provider, including relevant memory and storage.
                \item The paper should provide the amount of compute required for each of the individual experimental runs as well as estimate the total compute.
                \item The paper should disclose whether the full research project required more compute than the experiments reported in the paper (e.g., preliminary or failed experiments that didn't make it into the paper).
            \end{itemize}

        \item {\bf Code Of Ethics}
        \item[] Question: Does the research conducted in the paper conform, in every respect, with the NeurIPS Code of Ethics \url{https://neurips.cc/public/EthicsGuidelines}?
        \item[] Answer: \answerYes{} 
        \item[] Justification: The authors have reviewed the NeurIPS Code of Ethics. The research conducted in this paper conforms with it in every respect.
        \item[] Guidelines:
            \begin{itemize}
                \item The answer NA means that the authors have not reviewed the NeurIPS Code of Ethics.
                \item If the authors answer No, they should explain the special circumstances that require a deviation from the Code of Ethics.
                \item The authors should make sure to preserve anonymity (e.g., if there is a special consideration due to laws or regulations in their jurisdiction).
            \end{itemize}

        \item {\bf Broader Impacts}
        \item[] Question: Does the paper discuss both potential positive societal impacts and negative societal impacts of the work performed?
        \item[] Answer: \answerNA{} 
        \item[] Justification: This paper is mostly theoretical.
        \item[] Guidelines:
            \begin{itemize}
                \item The answer NA means that there is no societal impact of the work performed.
                \item If the authors answer NA or No, they should explain why their work has no societal impact or why the paper does not address societal impact.
                \item Examples of negative societal impacts include potential malicious or unintended uses (e.g., disinformation, generating fake profiles, surveillance), fairness considerations (e.g., deployment of technologies that could make decisions that unfairly impact specific groups), privacy considerations, and security considerations.
                \item The conference expects that many papers will be foundational research and not tied to particular applications, let alone deployments. However, if there is a direct path to any negative applications, the authors should point it out. For example, it is legitimate to point out that an improvement in the quality of generative models could be used to generate deepfakes for disinformation. On the other hand, it is not needed to point out that a generic algorithm for optimizing neural networks could enable people to train models that generate Deepfakes faster.
                \item The authors should consider possible harms that could arise when the technology is being used as intended and functioning correctly, harms that could arise when the technology is being used as intended but gives incorrect results, and harms following from (intentional or unintentional) misuse of the technology.
                \item If there are negative societal impacts, the authors could also discuss possible mitigation strategies (e.g., gated release of models, providing defenses in addition to attacks, mechanisms for monitoring misuse, mechanisms to monitor how a system learns from feedback over time, improving the efficiency and accessibility of ML).
            \end{itemize}

        \item {\bf Safeguards}
        \item[] Question: Does the paper describe safeguards that have been put in place for responsible release of data or models that have a high risk for misuse (e.g., pretrained language models, image generators, or scraped datasets)?
        \item[] Answer: \answerNA{} 
        \item[] Justification: This paper poses no such risks.
        \item[] Guidelines:
            \begin{itemize}
                \item The answer NA means that the paper poses no such risks.
                \item Released models that have a high risk for misuse or dual-use should be released with necessary safeguards to allow for controlled use of the model, for example by requiring that users adhere to usage guidelines or restrictions to access the model or implementing safety filters.
                \item Datasets that have been scraped from the Internet could pose safety risks. The authors should describe how they avoided releasing unsafe images.
                \item We recognize that providing effective safeguards is challenging, and many papers do not require this, but we encourage authors to take this into account and make a best faith effort.
            \end{itemize}

        \item {\bf Licenses for existing assets}
        \item[] Question: Are the creators or original owners of assets (e.g., code, data, models), used in the paper, properly credited and are the license and terms of use explicitly mentioned and properly respected?
        \item[] Answer: \answerNA{} 
        \item[] Justification: This paper does not use existing assets.
        \item[] Guidelines:
            \begin{itemize}
                \item The answer NA means that the paper does not use existing assets.
                \item The authors should cite the original paper that produced the code package or dataset.
                \item The authors should state which version of the asset is used and, if possible, include a URL.
                \item The name of the license (e.g., CC-BY 4.0) should be included for each asset.
                \item For scraped data from a particular source (e.g., website), the copyright and terms of service of that source should be provided.
                \item If assets are released, the license, copyright information, and terms of use in the package should be provided. For popular datasets, \url{paperswithcode.com/datasets} has curated licenses for some datasets. Their licensing guide can help determine the license of a dataset.
                \item For existing datasets that are re-packaged, both the original license and the license of the derived asset (if it has changed) should be provided.
                \item If this information is not available online, the authors are encouraged to reach out to the asset's creators.
            \end{itemize}

        \item {\bf New Assets}
        \item[] Question: Are new assets introduced in the paper well documented and is the documentation provided alongside the assets?
        \item[] Answer: \answerNA{} 
        \item[] Justification: This paper does not release new assets.
        \item[] Guidelines:
            \begin{itemize}
                \item The answer NA means that the paper does not release new assets.
                \item Researchers should communicate the details of the dataset/code/model as part of their submissions via structured templates. This includes details about training, license, limitations, etc.
                \item The paper should discuss whether and how consent was obtained from people whose asset is used.
                \item At submission time, remember to anonymize your assets (if applicable). You can either create an anonymized URL or include an anonymized zip file.
            \end{itemize}

        \item {\bf Crowdsourcing and Research with Human Subjects}
        \item[] Question: For crowdsourcing experiments and research with human subjects, does the paper include the full text of instructions given to participants and screenshots, if applicable, as well as details about compensation (if any)?
        \item[] Answer: \answerNA{} 
        \item[] Justification: This paper does not involve crowdsourcing nor research with human subjects.
        \item[] Guidelines:
            \begin{itemize}
                \item The answer NA means that the paper does not involve crowdsourcing nor research with human subjects.
                \item Including this information in the supplemental material is fine, but if the main contribution of the paper involves human subjects, then as much detail as possible should be included in the main paper.
                \item According to the NeurIPS Code of Ethics, workers involved in data collection, curation, or other labor should be paid at least the minimum wage in the country of the data collector.
            \end{itemize}

        \item {\bf Institutional Review Board (IRB) Approvals or Equivalent for Research with Human Subjects}
        \item[] Question: Does the paper describe potential risks incurred by study participants, whether such risks were disclosed to the subjects, and whether Institutional Review Board (IRB) approvals (or an equivalent approval/review based on the requirements of your country or institution) were obtained?
        \item[] Answer: \answerNA{} 
        \item[] Justification: This paper does not involve crowdsourcing nor research with human subjects.
        \item[] Guidelines:
            \begin{itemize}
                \item The answer NA means that the paper does not involve crowdsourcing nor research with human subjects.
                \item Depending on the country in which research is conducted, IRB approval (or equivalent) may be required for any human subjects research. If you obtained IRB approval, you should clearly state this in the paper.
                \item We recognize that the procedures for this may vary significantly between institutions and locations, and we expect authors to adhere to the NeurIPS Code of Ethics and the guidelines for their institution.
                \item For initial submissions, do not include any information that would break anonymity (if applicable), such as the institution conducting the review.
            \end{itemize}

    \end{enumerate}
\fi

\end{document}